\documentclass[reqno]{amsart}
\usepackage[english]{babel}

\usepackage{geometry}
\usepackage[round,comma]{natbib}                %

\usepackage{enumerate}
\usepackage{mathabx}

\usepackage{amsmath}
\usepackage{amsthm}
\usepackage{amsfonts}
\usepackage{amssymb}
\usepackage{dsfont}

\usepackage{nicefrac}
\usepackage{booktabs}
\usepackage{mathrsfs}
\usepackage{bm}
\usepackage{mathtools}

\newcommand{\ccF}{{\mathscr F}}\newcommand{\cF}{{\mathcal F}}
\newcommand{\ccG}{{\mathscr G}}\newcommand{\cG}{{\mathcal G}}

\newcommand{\cM}{{\mathcal M}}

\newcommand{\ccP}{{\mathscr P}}

\newcommand{\restr}{\mathbf{\kern0.3ex%
 \vert\kern-0.3ex}\backprime\kern0.3ex}

\renewcommand{\cF}{\ccF}
\renewcommand{\cG}{\ccG}

\newtheorem{theorem}{Theorem}[section]
\newtheorem{lemma}[theorem]{Lemma}              %
\newtheorem{proposition}[theorem]{Proposition}  %

\theoremstyle{definition}
\newtheorem{example}[theorem]{Example} %
\newtheorem{definition}[theorem]{Definition} %
\newtheorem{remark}[theorem]{Remark}%
\newtheorem{assumption}[theorem]{Assumption}%

\usepackage[backgroundcolor=white,bordercolor=orange]{todonotes}

\usepackage[colorlinks,urlcolor=red,citecolor=blue,linkcolor=red]{hyperref}

\DeclareMathOperator{\esssup}{ess  \sup}

\usepackage{subfiles} %

\renewcommand{\P}{P}

\DeclareMathOperator*{\esssupZ}{ess\,sup}

\DeclareMathOperator*{\esssupN}{ess\,sup}

\numberwithin{equation}{section}

\title{Robust asymptotic insurance-finance arbitrage}
\author[Oberpriller]{Katharina Oberpriller}
		\address{University of Freiburg, Ernst-Zermelo-Str. 1, 79104 Freiburg, Germany.}
		\email{ katharina.oberpriller@stochastik.uni-freiburg.de}
\author[Ritter]{Moritz Ritter}
		\address{University of Freiburg, Ernst-Zermelo-Str. 1, 79104 Freiburg, Germany.}
		\email{ moritz.ritter@stochastik.uni-freiburg.de}
\author[Schmidt]{Thorsten Schmidt}
		\address{University of Freiburg, Ernst-Zermelo-Str. 1, 79104 Freiburg, Germany.}
		\email{ thorsten.schmidt@stochastik.uni-freiburg.de}

 \date{\today. }
    \thanks{
    The authors gratefully acknowledge  the support of the German Research Foundation (DFG) through grant SCHM 2160/15-1 and of the Freiburg Center for Data Analysis and Modeling (FDM)}

\begin{document}
\maketitle

\begin{abstract}
In most cases, insurance contracts are linked to the financial markets, such as through interest rates or equity-linked insurance products. To motivate an evaluation rule in these hybrid markets, \cite{artzner2022insurance} introduced the notion of insurance-finance arbitrage. 
In this paper we extend their setting by incorporating model uncertainty. To this end, we allow statistical uncertainty in the underlying dynamics to be represented by a set of priors $\ccP$. 
Within this framework we introduce the notion of \emph{robust asymptotic insurance-finance arbitrage} and characterize the absence of such strategies in terms of the concept of ${Q}\ccP$-evaluations. This is a nonlinear two-step evaluation which guarantees \emph{no robust asymptotic insurance-finance arbitrage}. 
Moreover, the  ${Q}\ccP$-evaluation dominates all two-step evaluations as long as we agree on the  set of priors $\ccP$ which shows that those two-step evaluations do not allow for robust asymptotic insurance-finance arbitrages.

Furthermore, we introduce a doubly stochastic model under uncertainty for surrender and survival. In this setting, we describe conditional dependence by means of copulas and illustrate how the ${Q}\ccP$-evaluation can be used for the pricing of hybrid insurance products.

\bigskip

\noindent\textbf{Keywords:} insurance-finance arbitrage under uncertainty, robust $QP$-rule, enlargement of filtration, absence of robust insurance-finance arbitrage.
\end{abstract}

\section{Introduction}
This paper develops and characterizes the absence of insurance-finance arbitrage under model uncertainty. Our starting point is the observation that most insurance contracts are strongly linked to 
financial markets, such as through interest rates or  via  direct links of the contractual benefits to stocks or indices.

However, due to the varied characteristics of  insurance contracts -- which are static, personalized products -- and  financial markets -- where products are standardized and  traded frequently -- the approaches to model these markets and the corresponding valuation techniques are fundamentally distinct. In the literature, several different approaches have been proposed how to deal with both insurance and financial markets in a consistent manner (see e.g., \cite{dhaene2017fair}, \cite{Dhaenekukushetal-2013}, \cite{pelsser_stadje_2014} and  \cite{semyon_eugene_wuerthich_2008} and references therein).

The notion of insurance-finance arbitrage (IFA) was introduced  by \cite{artzner2022insurance} with the aim of identifying a suitable notion of arbitrage in these  markets. Characteristic of such markets is that, on the one hand, the insurance company may issue contracts to a large number of clients, yet, on the other hand, it is able to simultaneously hedge its positions by trading on the financial market. In order to model the two information flows that the insurance company has access to, we work with two filtrations. The smaller filtration represents the publicly available information on the financial market, denoted by $\mathbb{F}=(\ccF_t)_{t \leq T}$, while the larger filtration $\mathbb{G}=(\ccG_t)_{t \leq T}$ additionally contains the insurer's information.
Given a pricing measure $Q$ on the financial market $(\Omega, \ccF)$ and a statistical measure $P$ on $(\Omega, \ccG)$ a characterization of the absence of IFA in terms of the ${Q} P$-rule is derived in \cite{artzner2022insurance}.

Even if a large set of homogeneous data is available, statistical uncertainties in predicting the future evolution of insurance losses in the considered portfolio remain a  problem that needs to be addressed. Hence, in this work we aim to take this uncertainty into account.
To do so, we fix the nullsets $\mathcal{N}$ on the financial market $(\Omega, \ccF)$ which determines the set of equivalent martingale measures $Q$. Second, we consider a class of probabilistic models $\ccP$ on $(\Omega,\mathscr G)$  with the $\ccF$-nullsets $\mathcal{N}$ and study the associated ${Q}\ccP$-rule. %
Most notably, this framework allows us to model uncertainty on the insurance market under the assumption that we do not face any model risk on the financial market. 
Working with a class of potential models $\ccP$ is in line with the growing literature on model risk and uncertainty (see e.g., \cite{Martini}, \cite{peng2019nonlinear}, \cite{denis_hu_peng_2010}, \cite{bbkn_2017}, \cite{soner_touzi_zhang_MRP} and \cite{bouchard_nutz_2015}). 

To the best of our knowledge, this is the first study of insurance-finance arbitrage under model uncertainty. More specifically, we prove a characterization of robust insurance-finance arbitrage (RIFA) by using the ${Q} \ccP$-evaluation in Theorem \ref{theorem:characterizationRobustInsuranceFinance}. 

Furthermore, we show that this result provides a theoretical foundation for a class of two-step evaluations introduced in \cite{pelsser_stadje_2014}, which is applied for the pricing of hybrid products depending on the financial market, as well as on other random sources, e.g., insurance products.
In particular, we prove that every two-step evaluation, which is the combination of a risk-free measure $Q$ and a coherent $\ccF$-conditional risk measure, which is continuous from below, equals a $Q\ccP$-evaluation for a suitable subset $\ccP$ on $(\Omega, \ccG).$ Thus, every two-step evaluation of this kind leads to a robust arbitrage-free price in the asymptotic insurance-finance setting.

We conclude the paper by suggesting possible applications in Section \ref{sec:GeneralizedSetting}.
First, we consider the case of two conditional independent random times such that the conditional distribution functions face a certain degree of  uncertainty. Here, the random times represent the surrender time and the time of death of an insurance seeker. In this setting, we introduce a financial market via a Cox-Ross-Rubinstein model and consider finance-linked insurance benefits with surrender options. We also numerically provide the robust arbitrage-free price of these products. Most notably, we find that the uncertainty justifies  the robust arbitrage-free price being higher than the supremum of all arbitrage-free prices under each possible model.
We then generalize the framework by allowing some dependence structure between the random times described by a copula. By doing this, we show that the well-known Cox model under uncertainty is contained in the outlined setting.

The paper is structured as follows. In Section \ref{sec:RobustAsymptoticInsuranceFinanceArbitrage} we introduce the definition of a robust insurance-finance arbitrage and provide a characterization in our main result. After that, Section \ref{sec:TwoStepValuation} studies the relation of the $Q \ccP$-rule to two-step evaluations. Then, in Section \ref{sec:GeneralizedSetting} we study an insurance-finance market with two conditionally independent random times and numerically provide the robust arbitrage-free prices for certain hybrid products. In addition, we consider a copula framework for the two random times under uncertainty.

\section{Robust asymptotic insurance-finance arbitrage} \label{sec:RobustAsymptoticInsuranceFinanceArbitrage}
Let $(\Omega, \cG)$ be a measurable space and denote by $\mathfrak{P}(\Omega,\cG)$ the set of all probability measures on this space. We consider a discrete time model with times $t=0,...,T$ and introduce two different kinds of information flows described by the filtrations $\mathbb{F}$ and $\mathbb{G}$ on $(\Omega, \cG)$. The filtration $\mathbb{F}=(\ccF_t)_{t\leq T}$ represents \emph{publicly available information} and contains all information available on the financial market. The filtration $\mathbb{G}=(\ccG_t)_{t\leq T}$ contains additional private information of the  considered insurance company, which includes, for example, several datasets on its clients. In particular, we assume $\mathbb{F} \subseteq \mathbb{G}$. Moreover, let $\ccF_0=\ccG_0=\lbrace \emptyset, \Omega \rbrace$ and $\mathscr F$ be a $\sigma$-field  such that $\mathscr F_T\subseteq \mathscr F\subseteq \mathscr G$. 

Given $\ccP \subseteq \mathfrak{P}(\Omega, \cG)$, a set $A\subseteq \Omega$ is called \emph{$\ccP$-polar} if $A\subseteq N$ for some $N\in \cG$ satisfies $P(N)=0$ for all $\P\in\ccP$. Moreover, a property holds \emph{$\ccP$-quasi surely} ($\ccP$-q.s.) if it holds outside a $\ccP$-polar set. \\
For  a fixed $\sigma$-ideal of $\ccF$-nullsets $\mathcal N$, i.e., there exists a measure $P_0$ on $(\Omega,\ccF)$ such that $\mathcal{N}$ are the nullsets of $P_0$, we define the following sets of priors
\begin{align} \label{eq:PriorsEquivalenToP_1}
    \mathfrak{P}_{\mathcal{N}}(\Omega, \ccF)\coloneqq \big\lbrace P \in \mathfrak{P}(\Omega,\ccF) \, \vert\,\mathcal N \text{ are the nullsets of }P  \big\rbrace
\end{align}
and
\begin{align} \label{eq:PriorsEquivalenToP_2}
    \mathfrak{P}_{\mathcal{N}}(\Omega, \ccG)\coloneqq \big\lbrace P \in \mathfrak{P}(\Omega,\ccG) \, \vert\,\mathcal N \text{ are the nullsets of }P\vert_{\mathscr F} \big\rbrace.
\end{align}
Hereafter, we fix a probability measure $P_0 \in \mathfrak{P}(\Omega, \ccF)$, which determines the nullset $\mathcal N$ on $(\Omega,\mathscr F)$, and a subset of priors $\ccP \subseteq  \mathfrak{P}_{\mathcal{N}}(\Omega, \ccG)$, which specifies the uncertainty about the model.
\\

In the following, we introduce the concept of insurance-finance arbitrage. For simplicity, we consider only a single insurer. We also assume that the insurer is able to contract with an arbitrarily high number of clients so as to reduce its risks. To establish this, we consider a finite number of insurance seekers and treat the limits of portfolio allocations - a technique inspired by large financial markets (e.g.,~\cite{Klein2000,kabanov1995large,KleinSchmidtTeichmann2015}). Such strategies may lead to \emph{insurance arbitrages} and it is partly our aim to characterize when and how such arbitrages can be achieved and under which conditions they can be avoided.

At the same time, the insurance company also trades on the financial market, potentially leading to a \emph{financial arbitrage}. Thus, the combination of these concepts results in an \emph{insurance-finance arbitrage.}

\subsection{The insurance contracts}
Insurance contracts offer a variety of benefits at future times in exchange for a single premium or a premium stream. We work with discounted quantities and, without loss of generality, we consider  a single premium paid at time 0 and an aggregated benefit received at future time $T$. More precisely, we denote by $p \in \mathbb{R}$  the premium to be paid at time $0$ and
 a $\ccG_T$-measurable benefit $X^i$   to be received by the $i$th client at time $T$. This allows to cover a wide range of contracts, particularly contracts depending on financial markets, such as  variable annuities.

 We assume that all insurance seekers under consideration can be treated as  homogeneous (under each model $P\in \ccP$)  and each insurance seeker  pays the same premium $p$ in order to receive his or her personal  benefit $X^i$.
This idea is formalized in the following assumption, which is a generalization of the framework of actuarial mathematics to stochastic assets, as discussed in \cite{artzner2022insurance}.
\begin{assumption} \label{assum:ConditionalExpectations}
For all $P \in \ccP$, the following holds:
\begin{enumerate}
    \item $X^1,X^2,... \in L_+^2(\Omega, \cG, P)$ are $\cF$-conditionally independent.
    \item $E_P[X^i\vert \cF]= E_P[X^1\vert \cF]$ for all $i \in \mathbb{N}$.
    \item $\textnormal{Var}_P[X^i \vert \cF]=\textnormal{Var}_P[X^1 \vert \cF]$ for all $i \in \mathbb{N}$.
\end{enumerate}
\end{assumption}
The insurance portfolio is obtained as a limit of \emph{allocations} of contracts with a finite number of clients. 
An allocation at time $0$, $\psi=(\psi^i)_{i \in \mathbb{N}},$ is $c_{00}$-valued, deterministic and non-negative, where $c_{00}$ denotes the space of sequences with a finite number of non-zero elements. For $i \in \mathbb{N}$, $\psi^i \in \mathbb{R}_+$ denotes the size of the contract with the $i$th insurance client. The accumulated benefits and premiums associated with the allocation $\psi$ are given by $\sum_{i \in \mathbb{N}} \psi^i X^i$ and $\sum_{i \in \mathbb{N}} \psi^i p$, while the associated profits and losses are denoted by
$$
V_T(\psi):= \sum_{i \in \mathbb{N}} \psi^i (p-X^i).
$$
An \emph{insurance portfolio strategy} $\Psi:=(\psi^n)_{n \in \mathbb{N}} \in c_{00}^{\mathbb{N}}$ is modeled as a sequence of allocations $\psi^n=(\psi^{n,i})_{i \in \mathbb{N}}$. Moreover, the profit and loss of an insurance portfolio strategy $\Psi$ is given (if it exists) by
$$
V_T(\Psi):=\lim_{n \to \infty} V_T(\psi^n)=\lim_{n \to \infty}  \sum_{i \in \mathbb{N}} \psi^{n,i} (p-X^i)
$$
We introduce the following \emph{admissibility conditions} for an insurance portfolio strategy $\Psi=(\psi^n)_{n \in \mathbb{N}}$.
\begin{assumption} \label{assum:UniformBoundedness}
\begin{enumerate}
    \item \emph{Uniform boundedness:} There exists $C>0,$ such that
    \begin{equation*} 
     \| \psi^n\|:= \sum_{i \geq 1} \psi^{n,i} \leq C \text{ for all } n \geq 1.
    \end{equation*}
    \item \emph{Convergence of the total mass:} There exists $ \gamma \geq 0$ such that
    \begin{equation*} 
         \| \psi^n\| \to \gamma.
    \end{equation*}
    \item \emph{Convergence of the total wealth:} There exists a $\mathbb{R}$-valued random variable $V_T(\Psi),$ such that 
    \begin{equation*}
       V_T(\Psi)= \lim_{n \to \infty} V_T(\psi^n) \quad \ccP\text{-q.s.}
    \end{equation*}
\end{enumerate}
\end{assumption}
An insurance portfolio strategy $\Psi=(\psi^n)_{n \in \mathbb{N}}$ which satisfies Assumption \ref{assum:UniformBoundedness} is called $\ccP$-\emph{admissible}.

\subsection{The financial market} \label{section:FinancialMarket}
We introduce a financial market model in discrete time consisting of $d+1$ assets $S=(S_t^0,...,S_t^d)_{t=0,...,T}$ on $(\Omega, \ccF).$ For $i=1,...,d,$ the discounted price of the $i$th asset at time $t$ is given by the $\cF_t$-measurable random variable $S_t^i$ and the bank account is given by $S^0\equiv 1.$ 
We assume that the insurance company trades with $\mathbb{F}$-trading strategies on the financial market, where a $\mathbb{F}$-trading strategy is a $d$-dimensional $\mathbb{F}$-predictable process $\xi=(\xi_t)_{t=1,...,T}$ with $\xi_t=(\xi_t^1,...,\xi_t^d).$
Note that each strategy $\xi$ can be extended by $\xi^0=(\xi_t^0)_{t=1,...,T}$ to a self-financing trading strategy $\bar{\xi}=(\xi_t^0,...,\xi_t^d)_{t=1,...,T}$, cf.\@ \cite[Remark 5.8]{follmer2016stochastic}. The associated \emph{gain process} at time $t=1,...,T$ is then given by the discrete stochastic integral
$$
(\xi \cdot S)_t\coloneqq\sum_{s=1}^t \xi_s \left(S_s -S_{s-1}\right)= \sum_{s=1}^t \sum_{i=1}^d\xi_s^i \left(S_s^i -S_{s-1}^i\right).
$$
As is typical, the absence of arbitrage in this market with respect to any restricted measure $P\vert_{\mathscr F}$ for $P\in\ccP\subseteq \mathfrak P_{\mathcal N}(\Omega,\ccG)$ can be characterized by the existence of an equivalent martingale measure, cf.\@ \cite[Theorem 5.16]{follmer2016stochastic}. The set of all equivalent martingale measures is denoted by $\cM_e(\mathbb{F})$ and defined by
\begin{equation} \label{eq:NoArbitrageMeasuresFinancialMarketF}
\cM_e(\mathbb{F}):= \lbrace Q \in \mathfrak{P}_{\mathcal{N}}(\Omega, \cF) \vert \   S \text{ is a } (\mathbb Q,\mathbb F)\text{-martingale}\rbrace.
\end{equation}

\begin{remark}\label{rem:uncertainty}
Note that we consider uncertainty in a narrow sense  on the financial market $(\Omega, \ccF)$ by taking into account a set of priors $\ccP \subseteq \mathfrak{P}_{\mathcal N}(\Omega, \ccG)$. More specifically, according to the definition of $\mathfrak{P}_{\mathcal{N}}(\Omega, \ccG)$ in \eqref{eq:PriorsEquivalenToP_2}, all nullsets on $(\Omega, \ccF)$ are already determined by $\mathcal{N}.$ The uncertainty of our model refers only to the additional insurance part defined on $(\Omega,\mathscr G)$. 
\end{remark}

\subsection{Robust insurance-finance arbitrage}
Finally, we introduce the \emph{insurance-finance market} $(S,\mathcal X,p)$ on  $(\Omega, \ccG)$ consisting of the benefits $\mathcal X=(X^i)_{i \in \mathbb{N}}$, the premium $p$, and the discounted prices of the assets $S=(S_t)_{t=0,...,T}.$

In order to define a robust arbitrage in our setting, we use the concept of a robust arbitrage strategy introduced in \cite{bouchard_nutz_2015}. Note, however, that we do not need to require convexity of $\ccP$.

\begin{definition} \label{def:RobustInsuranceFinance}
A $\ccP$\emph{-robust asymptotic insurance-finance arbitrage} $\textnormal{RIFA}(\ccP)$ on the insurance-finance market $(S,\mathcal X,p)$ is a pair $(\xi,\Psi)$ consisting of an $\mathbb{F}$-predictable trading strategy $\xi$ and a $\ccP$-admissible insurance portfolio strategy $\Psi$ such that
	\begin{align} \label{eq:RobustInsuranceFinance}
		(\xi\cdot S)_T+V_{T}(\Psi) \geq 0 \  \mathcal \ccP\text{-q.s.\@} \quad\text{and}\quad E_{\P}\left[	(\xi\cdot S)_T+V_{T}(\Psi) \right]>0\ \text{for some } \P\in\ccP.
	\end{align}
	If there exists no such pair $(\xi,\Psi)$ satisfying \eqref{eq:RobustInsuranceFinance}, then there is \emph{no $\ccP$-robust asymptotic insurance-finance arbitrage}, which we denote by $\textnormal{NRIFA}(\ccP).$ 
\end{definition}

\begin{remark} \label{remark:NoRobustArbitrageSingleMeasureArbitrage}
If for all $P \in \ccP$ it holds $\textnormal{NRIFA}(\{P\})$ then  $\textnormal{NRIFA}(\ccP)$ is also satisfied. However, the converse statement is incorrect. 
\end{remark}

In the case of no model uncertainty, i.e.,\@ $\ccP=\lbrace P \rbrace$ for a measure $P \in \mathfrak{P}_{\mathcal{N}}(\Omega, \ccG)$, it is shown in Corollary 5.2 in \cite{artzner2022insurance} that the following relation to the ${Q} P$-rule holds. If there exists ${Q} \in \cM_e(\mathbb{F})$ such that 
\begin{equation} \label{eq:ConditionPremiumNoArbitrage}
    p \leq E_{{Q} \odot P}[X^1]\coloneqq E_{{Q}}\left[ {E}_P[X^1 \vert \ccF] \right],
\end{equation}
then there exists no asymptotic insurance-finance arbitrage. Thus, according to Remark \ref{remark:NoRobustArbitrageSingleMeasureArbitrage}, an insurance-finance market $(S,\mathcal X,p)$ fulfills $\textnormal{NRIFA}(\ccP)$ if (\ref{eq:ConditionPremiumNoArbitrage}) holds for each $P\in \ccP$.  However, it should be noted that these conditions are not strictly necessary as we will demonstrate in the following.

\subsection{Uniform essential supremum}\label{sec:uniformesssup}
In order to characterize the absence of a $\ccP$-robust asymptotic insurance-finance arbitrage, we aim to identify assumptions that allow us to take into account a robust version of the conditional expectations $E_P[X \vert \ccF]$ for $P \in \ccP$ in \eqref{eq:ConditionPremiumNoArbitrage}. Given a set of priors  $\ccP \subseteq \mathfrak{P}(\Omega, \ccG)$, it is in general not possible to consider ``$\sup_{P \in \ccP} E_P[X \vert \ccF]$'', as the conditional expectation is only defined $P$-a.s. and the priors in $\ccP$ may have different nullsets. Furthermore, the supremum no longer needs to be measurable.

 The natural approach to solving this measurability issue is to work with the essential supremum instead of the supremum.
However, for a general set of priors $\ccP \subseteq \mathfrak{P}(\Omega, \ccG)$ and a set of random variables $\Phi$ on $(\Omega, \mathscr G)$ there may be no uniform essential supremum, i.e.,\@ a random variable $Y$ such that 
\begin{equation}\label{eq:uniformesssup}
   Y= \esssup^{P}\Phi\ \quad  P\text{-a.s.\@ for all }P\in \mathscr P.  
\end{equation}
We refer to \cite{follmer2016stochastic} for a precise definition of the essential supremum and to \cite{bartlConditionalExpectation_2020}, in whose work a general construction of such a nonlinear conditional expectation is studied.

The definition of the set $\mathfrak{P}_{\mathcal{N}}(\Omega, \mathscr G)$  in \eqref{eq:PriorsEquivalenToP_2} allows us to consider for $\ccP\subseteq \mathfrak{P}_{\mathcal{N}}(\Omega, \mathscr G)$ the set of random variables $\Phi=(\varphi^P)_{P\in \ccP}$ such that
\begin{equation}
    \varphi^P=E_P[X\vert \mathscr F]\quad \P\text{-a.s.\@ for all } P\in  \ccP.
\end{equation}
Indeed, we observe that by $\ccF$-measurability together with the fact that $P\vert_{\mathscr F}\sim P'\vert_{\mathscr F}$ for all $P,P'\in \ccP,$ the conditional expectation $E_{P}[X\vert \mathscr F]$ is not only $P$-a.s. uniquely determined but also $P'$-a.s. for all $P'\in \ccP$. Then, there exists a uniform essential supremum which fulfills (\ref{eq:uniformesssup}),  as the following result demonstrates.

\begin{lemma}\label{lem2.6}
	Let $\mathcal{N}$ be the nullsets generated by the probability measure $P_0 \in \mathfrak{P}(\Omega, \ccF)$, $\ccP \subseteq \mathfrak{P}_{\mathcal N}(\Omega, \ccG)$ and $\Phi$ a set of $\ccF$-measurable random variables. Then 
	\begin{align} \label{eq:EssupNullsetMathcalN}
		\esssup^{ P_0}\Phi=\esssup^{ P}\Phi \quad \P\text{-a.s.\@ for all } P\in  \ccP. 
	\end{align}
\end{lemma}
\begin{proof} We show that $\esssup^{ P_0}\Phi$ fulfills condition $i)$ and $ii)$ in \cite[Theorem A.37]{follmer2016stochastic} on $(\Omega,\mathscr G, P_0)$ for all $P\in  \ccP$.
	Using the definition of $\esssup^{ P_0} \Phi$, we geht the following:
	\begin{align}\label{eq:essentialsupremum1}
		\esssup^{ P_0} \Phi\geq \varphi^P \quad P_0\text{-a.s.\@ for all }\varphi^P\in\Phi.
	\end{align}
Since $\{\esssup^{ P_0} \Phi\geq \varphi^P\}$ is $\ccF$-measurable and $\P \vert_{\ccF}\sim P_0$ for all $\P\in \ccP$, \eqref{eq:essentialsupremum1} also holds $P$-a.s. for any $P\in\ccP$.
Relying on the construction of the essential supremum, see, for example \cite[Theorem A.37]{follmer2016stochastic}, there exists a countable subset $\Phi^*\subset \Phi$ such that 
\begin{align*}
	\esssup^{ P_0} \Phi(\omega)=	\sup \Phi^*(\omega)\quad \text{ for all }\omega\in \Omega.
\end{align*}
Fix any $P \in \ccP$, then for each random variable $\psi$ on $(\Omega,\ccG)$ such that 
\begin{align*}
	\psi \geq \phi \quad P \text{-a.s.},
\end{align*}
we get $\psi\geq 		\sup \Phi^*=\esssup^P \Phi\ \P$-a.s. 
\end{proof}

In the following, for the fixed  measure $P_0$, which generates the $\mathscr F$-nullsets $\mathcal N$, we use the notation
$$
 \esssup \Phi :=\esssup^{ P_0}\Phi.
$$
\subsection{Characterization of robust insurance-finance arbitrage}
In the next result we use the essential supremum defined in Section \ref{sec:uniformesssup} to characterize $\textnormal{NRIFA}(\ccP).$ Furthermore, to prove this theorem, we use a measure extension ${Q}\odot P$ which was studied for the first time in \cite{plachky_rueschendorf_1984} (see also Proposition 4.1 in \cite{artzner2022insurance}). For the reader's convenience we state the existence result for such a measure.
\begin{proposition}\label{prop: QPmeasure}
    Let $(\Omega, \ccG,P )$ be a probability space, $\ccF \subseteq \ccG$ a sub-$\sigma$-algebra, and $Q\sim P\vert_{\mathscr F}$ a probability measure on $(\Omega, \ccF)$. Then, there exists a unique probability measure, denoted by ${Q} \odot P,$ on $(\Omega, \ccG)$ such that ${Q}\odot P={Q}$ on $\ccF$ and ${Q} \odot P=P$ conditioned on $\mathscr F$. 
 \end{proposition}
\begin{remark}\label{rem: definitionQP}
The measure ${Q} \odot P$ on $(\Omega, \ccG)$ from Proposition \ref{prop: QPmeasure} can also be characterized by each of the following more explicit expressions: 
\begin{enumerate}
    \item For all $G\in\mathscr G$ we have  ${Q} \odot P[G]= E_Q[E_P[\mathds{1}_G\vert\mathscr F]]$. 
    \item The density of  ${Q} \odot P$ with respect to $P$ is given by $dQ/dP\vert_{\mathscr F}$.
\end{enumerate}
This also shows that the measures  ${Q} \odot P$ and $P$ are equivalent.
\end{remark}
As is well-known, the arbitrage-free prices of a $\mathscr F_T$-measurable claim $H\geq 0$ in the financial market $S$ can be described by the set of expectations under all $Q\in\mathcal M_e(\mathbb F)$ such that $H\in L^1(Q)$, cf.\@ \cite[Theorem 5.29]{follmer2016stochastic}. The following Lemma shows that if we fix $P\in\ccP$ we could even only use the measures $Q\in\mathcal M_e(\mathbb F)$ such that $H\in L^1(Q)$ and such that $Q$ has bounded density with respect to $P$ to determine the set of arbitrage-free prices for $H$.

\begin{lemma}\label{lem: bounded martingale measures}
    Let $H\geq 0$ be a discounted $\mathscr F_T$-measurable claim such that $H\in L^1(Q)$ for some $Q\in\mathcal M_e(\mathbb F)$ and let $P\in\mathfrak{P}_{\mathcal{N}}(\Omega, \ccG)$. Then, there exists $Q^*\in\mathcal M_e(\mathbb F)$ such that $H\in L^1(Q^*),\ E_{Q}[H]=E_{Q^*}[H]$ and $dQ^*/dP\vert_{\mathscr F}\in L^{\infty}(\mathscr F,P\vert_{\mathscr F}).$
\end{lemma}
\begin{proof}
We define the process $S^{d+1}=(S_t^{d+1})_{t=0,...,T}$ by  $S^{d+1}_t=E_{Q}[H\vert \mathscr F_t]$. Then $Q$ is a martingale measure for the extended market $(S,S^{d+1})$. According to \cite[Theorem 5.16]{follmer2016stochastic}, it follows that this extended market is free of arbitrage and that there is an equivalent measure $Q^*\sim P\vert_{\mathscr F}$ such that $(S,S^{d+1})$ is a $(Q^*,\mathbb{F})$-martingale and $dQ^*/dP\vert_{\mathscr F}\in L^{\infty}(\mathscr F,P\vert_{\mathscr F}).$ The martingale property for $S$ implies $Q^*\in\mathcal M_e(\mathbb F)$ and for $S^{d+1}$ it implies $H=S_T^{d+1}\in L^1(Q^*)$ and  $E_{Q}[H]=S_0^{d+1}=E_{Q^*}[H]$.
\end{proof}

\begin{lemma}\label{lem: P help measure}
    Let $A_1,\dots,A_N\in\mathscr F$ be a partition of $\Omega$ and $P_1,\dots,P_N\in \mathfrak{P}_{\mathcal{N}}(\Omega, \ccG)$. Then the measure $P$ on $(\Omega,\mathscr G)$ defined by 
    \begin{align}\label{eq: P help measure}
        P(G)\coloneqq c^{-1}\sum_{i=1}^NP_i(G\cap A_i),
    \end{align}
    for $c=\sum_{i=1}^NP_i(A_i)$ is well-defined, a probability measure and $P\in\mathfrak{P}_{\mathcal{N}}(\Omega, \ccG)$.
\end{lemma}
\begin{proof}
    We must show that $c>0$. For the sake of contradiction we assume that $P_i[A_i]=0$ for all $i=1,\dots,N$. By the equivalence  $P_i\in \mathfrak{P}_{\mathcal{N}}(\Omega, \ccG)$ and $A_i\in\mathscr F$, it follows that $P_1[A_i]=0$ for all $i=1,\dots, N$ and thus $P_1[\Omega]=0$. This is a contradiction and we obtain $c>0$. Thus, $P$ is well-defined and a probability measure on $(\Omega,\mathscr G)$. In order to prove $P\in\mathfrak{P}_{\mathcal{N}}(\Omega, \ccG)$ we show that for all $F\in\mathscr F$ it holds that  $P_1(F)=0$ if and only if $P[F]=0$. First, we assume that $P_1[F]=0$.  Then, once more using the equivalence $P_i\in \mathfrak{P}_{\mathcal{N}}(\Omega, \ccG)$ we obtain $P_i[F\cap A_i ]\leq P_i[F]=0$ for all $i=1,\dots,N$ and it follows that $P[F]=0$. On the other hand, if $P[F]=0$, then it follows that $P_i[A_i\cap F]=0$ and thus $P_1[F]=0$. 
\end{proof}
Next, we state the main result of this paper.

\begin{theorem} \label{theorem:characterizationRobustInsuranceFinance}
	There is no $\ccP$-robust asymptotic insurance-finance arbitrage on $(S,\mathcal X,p)$ if and only if one of the following statements is fulfilled.
	\begin{enumerate}
		\item For all $ P \in \ccP$ there is a measure $Q\in \cM_e(\mathbb F)$ such that 
		$$ p\leq E_{Q}\left [E_{P}[X^1\vert \ccF]\right].$$
		\item There is a measure $Q \in \cM_e(\mathbb F)$ such that
		$$ p < E_{Q}\Big[\esssupN_{ P \in \ccP}E_{P}[X^1\vert \ccF]\Big].$$
    \end{enumerate}
\end{theorem}

\begin{proof}
We start with the \emph{only if} direction and prove this using contraposition. Assume that neither $(i)$ nor $(ii)$ are fulfilled, i.e.,
	\begin{enumerate}
		\item[$\overline{(i)}$] There is  $P'\in\ccP$ such that for all $ Q\in \cM_e(\mathbb F)$ it holds that
		$$ p> E_{Q}\left[E_{P'}[X^1\vert \ccF]\right].$$
	\end{enumerate}
	and \smallskip
	\begin{enumerate}
		\item[$\overline{(ii)}$] For all $Q\in\cM_e(\mathbb F)$ we have 
		$$ p \geq E_{Q} \Big[\esssupN_{P\in\ccP}E_{P}[X^1\vert \ccF] \Big].$$
	\end{enumerate}
We have to show that there is a $\ccP$-robust asymptotic insurance-finance arbitrage on $(S,\mathcal X,p)$.
Fix some $P'$ which satisfies $\overline{(i)}$ and define the set $A$ by 
$$A\coloneqq \Big \lbrace {E_{P'}[X^1\vert \cF]<\esssupN_{\P\in\ccP}E_{P}[X^1\vert  \ccF] \Big \rbrace} \in \ccF.$$ 
First, we consider the case $\P'[A]>\nolinebreak0$. Note that $Y:=\esssup_{\P\in\ccP}E_{P}[X^1\vert  \ccF]$ can be interpreted as a contingent claim on $(\Omega, \ccF)$. By $\overline{(ii)}$ we know that $p \geq \sup_{Q \in \cM_e(\mathbb{F})} E_{Q}[Y].$ This means $p$ is greater or equal than every arbitrage-free price for $Y$. Thus, as shown by \cite[Theorem 5.29, Corollary 7.9]{follmer2016stochastic}, there exists a $\mathbb{F}$-predictable (super-hedging) strategy $\xi$ such that
\begin{align}\label{eq:superhedging1}
	p+(\xi\cdot S)_T\geq \esssupN_{\P\in\ccP}E_{P}[X^1\vert \ccF]\quad P_0\text{-a.s.}
\end{align}
As both sides in \eqref{eq:superhedging1} are $\ccF$-measurable and by the definition of $\ccP$ Equation \eqref{eq:superhedging1} holds $P$-a.s. for all $P \in \ccP$. We define the admissible strategy $\Psi=(\psi^n)_{n\in\mathbb N}$ by  
\begin{align} \label{eq:DefinitionInsuranceStrategy}
    \psi^{n}\coloneqq \dfrac{1}{n}(1,\dots, 1, 0,\dots)=\dfrac{1}{n}(\mathds{1}_{\{k\leq n\}})_{k\in \mathbb N}.
\end{align} According to Assumption \ref{assum:ConditionalExpectations} and  \cite[Theorem 3.5]{Majerek2005}, we have 
$$ V_T(\Psi)=p-E_{P}[X^1\vert \ccF]\quad P\text{-a.s. for all } P\in\ccP.$$
The strategy $(\xi,\Psi)$ is a robust asymptotic insurance-finance arbitrage since by \eqref{eq:superhedging1} and the argumentation below it holds for all $P\in\ccP$ that
\begin{align*}
	(\xi\cdot S)_T+ V_T(\Psi)&=p+(\xi\cdot S)_T+ V_T(\Psi)-p\\
	&\geq \esssupN_{\P\in \ccP}E_{P}[X^1\vert \ccF]-E_{P}[X^1\vert \ccF] \quad P \text{-a.s.}\\
	&\geq 0\quad P \text{-a.s.} 
\end{align*}
Moreover, since $P'[A]>0$, we have 
\begin{align}
\esssupN_{\P\in \ccP}E_{P}[X^1\vert \cF]- E_{P'}[X^1\vert \cF] \in L_+^0\backslash \lbrace 0 \rbrace.
\end{align}
Second, we consider $P'[A]=0$. As $P' \vert_{\ccF}  \sim Q$ for all $Q \in \cM_e(\mathbb{F})$ and $A \in \ccF$, it also holds $Q[A]=0$ and consequently $Q[A^c]=1$. Then $\overline{i)}$ yields that for all $Q \in \cM_e(\mathbb{F})$
\begin{align}
 p> E_{Q}\left[ E_{P'}[X^1 \vert \ccF]\right]= E_{Q}\Big[ \esssupN_{\P\in \ccP} E_{P}[X^1 \vert \ccF]\Big]. 
 \end{align}
 In this case $p$ is strictly greater than any arbitrage-free price for $ \esssupN_{\P\in \ccP} E_{P}[X^1 \vert \ccF]$ and thus also by \cite[Theorem 5.29, Corollary 7.9]{follmer2016stochastic} we can find a $\mathbb{F}$-predictable strategy $\xi$ such that
 \begin{align}
	p+(\xi\cdot S)_T- \esssupN_{\P\in\ccP}E_{P}[X^1\vert \ccF]\in L^0_+\backslash\{0\} \text{ for all } P\in \ccP.
\end{align}
By once again choosing the insurance strategy $\Psi$ given in \eqref{eq:DefinitionInsuranceStrategy}, the pair $(\xi, \Psi)$ is a $\ccP$-robust asymptotic insurance-finance arbitrage as it holds for all $P\in\mathcal P$ that
\begin{align*}
	(\xi\cdot S)_T+ V_T(\Psi)&=p+(\xi\cdot S)_T+ V_T(\Psi)-p\\
	&=p+ (\xi \cdot S)_T -E_{P}[X^1\vert \ccF]\quad P\text{-a.s.}\\
	&\geq p+ (\xi \cdot S)_T -\esssupN_{\P\in \ccP} E_{P}[X^1 \vert \ccF]\in L^0_+\backslash\{0\}.
\end{align*}
This concludes the proof of the \emph{only if} direction and we proceed with the \emph{if} direction. If $(i)$ holds, then for any $P\in \ccP$ we could assume that the martingale measure $Q$ has bounded density with respect to   $P\vert_{\mathscr F}$, cf.\@  Lemma \ref{lem: bounded martingale measures}. According to Corollary 5.2 in \cite{artzner2022insurance}, there exists no $\{P\}$-robust asymptotic insurance-finance arbitrage for each fixed measure $P \in \ccP$ and thus $\textnormal{NRIFA}(\ccP)$ holds. For the reader's convenience we provide a detailed proof of this statement. Fix $P \in \ccP$. For the sake of contradiction assume that there is an admissible strategy $(\xi,\Psi)$ such that 
\begin{align*}
(\xi\cdot S)_T+V_T(\Psi)\in L_+^0(P)\backslash\{0\}.
\end{align*}
We define the process $Z=(Z_t)_{t\leq T}$ by
\begin{align}\label{eq:Zprocess}
    Z_0=0\quad\text{ and }\quad Z_t=E_{ Q\odot P}[V_T(\Psi)\vert \mathscr F_t] \quad \text{ for all }  1\leq t\leq T.
\end{align}
Using \cite[Proposition B.1]{artzner2022insurance} and assumption $(i)$ we obtain the following:
\begin{align*}
E_{ Q\odot P}[Z_t]&=E_{ Q\odot P}[V_T(\Psi)]\\
&=E_{ Q\odot P}[\gamma(p-X^1)]\\
&=\gamma(p-E_{ Q\odot P}[X^1])\leq 0.
\end{align*}
This shows that $Z$ is a local $(\mathbb F, Q\odot P)$-supermartingale. By \cite[Remark 9.5]{follmer2016stochastic} the value process $(\xi\cdot S)$ is a local  $(\mathbb F, Q\odot P)$-martingale and thus $(\xi\cdot S)+Z$ is a local $(\mathbb F, Q\odot P)$-supermartingale and fulfills 
\begin{align*}
(\xi\cdot S)_T+Z_T=(\xi\cdot S)_T+V_T(\Psi)\geq 0. 
\end{align*}    
By \cite[Proposition 9.6]{follmer2016stochastic} the process  $(\xi\cdot S)+Z$ is a $(\mathbb F, Q\odot P)$-supermartingale  and thus 
\begin{align*}
		E_{ Q\odot\P}[(\xi\cdot S)_T+V_{T}(\Psi)]\leq E_{ Q\odot P}[(\xi\cdot S)_0+Z_0]= 0. 
\end{align*}
According to Remark \ref{rem: definitionQP} it holds $ Q\odot P\sim P$ and we find a contradiction. This shows that $(i)$ implies $\textnormal{NRIFA}(\ccP)$. 
	
Let us now assume that $(ii)$ holds. In a first step, we show that there exists a finite partition of $\Omega$ given by $A_1,...,A_{N} \in \ccF$ and measures $(P^i)_{i \leq N}$ with $P^i \in \ccP$ such that
$$
p<\sum_{i=1}^{N} E_Q \left[ \mathds{1}_{A_i} E_{P^i}[X^1\vert  \ccF] \right].
$$
Based on the definition of the essential supremum, cf. \cite[Theorem A.37]{follmer2016stochastic}, there exists a countable subset $\ccP^*=(P^i)_{i \in \mathbb{N}} \subset \ccP$ such that
\begin{align} \label{eq:EssentialAsLimit}
\esssupN_{\P\in \ccP}E_P[X^1 \vert \ccF]= \sup_{P \in \ccP^{*}}E_P[X^1 \vert \ccF]= \lim_{N \to \infty} \sup_{P \in (P^i)_{i \leq N}} E_P[X^1 \vert \ccF] \quad P_0 \text{-a.s.}
\end{align}
Moreover, as the essential supremum is $\ccF$-measurable, \eqref{eq:EssentialAsLimit} also holds $\ccP$-q.s. and $Q$-a.s. for all $Q \in \cM_e(\mathbb{F}).$ 
\raggedbottom
Thus, using $(ii)$ and monotone convergence, we get the following:
\begin{align*}
     p &< E_{Q}\Big[\esssupN_{ P \in \ccP}E_{P}[X^1\vert \ccF]\Big]\\ 
     &=E_{Q} \Big[ \lim_{N \to \infty} \sup_{P \in (P^i)_{i \leq N}} E_P[X^1 \vert \ccF]\Big]\\
     &= \lim_{N \to \infty} E_{Q} \Big[  \sup_{P \in (P^i)_{i \leq N}} E_P[X^1 \vert \ccF]\Big].
\end{align*}
Thus, there exists $N \in \mathbb{N}$ such that 
$$
p < E_{Q} \Big[  \sup_{P \in (P^i)_{i \leq N}} E_P[X^1 \vert \ccF]\Big].
$$
For every $\ell \in 1,...,{N}$ we define the set $B_{\ell}$ as
\begin{align*}
	B_\ell\coloneqq  \Big \lbrace{ E_{P^\ell}[X^1\vert \ccF ]= \sup_{P\in(P^i)_{i\leq N}}E_{P}[X^1\vert  \ccF] \Big\rbrace}\in\ccF.
\end{align*}
and $A_1\coloneqq B_1$ and $A_\ell := B_\ell\backslash\bigcup_{k<\ell}B_k$ for $\ell=2,\dots,N$.
Then $(A_{\ell})_{\ell=1,...,N}$ are disjoint and build a partition of $\Omega$ because for each $\omega\in \Omega$ there is $\ell\in\{1,\dots, N\}$ such that 
\begin{align*}
    E_{P^\ell}[X^1\vert \ccF ](\omega)= \sup_{P\in(P^i)_{i\leq N}}E_{P}[X^1\vert  \ccF](\omega)
\end{align*}
and thus we obtain
\begin{align*}
 \biguplus_{\ell =1}^{N} A_{\ell}= \bigcup_{\ell =1}^{N} B_{\ell}=\Omega.
\end{align*}
Moreover, we have
\begin{align} \label{eq:pSmaller}
    p < E_{Q} \Big[ \sum_{i=1}^{N} \textbf{1}_{A_i} \sup_{P\in(P^i)_{i\leq N}}E_{P}[X^1\vert  \ccF] \Big]= \sum_{i=1}^{N} E_{Q} \left[ \textbf{1}_{A_i} E_{P^i}[X^1\vert  \ccF] \right].
\end{align}
We define the measure $P'\in\mathfrak{P}_{\mathcal{N}}(\Omega, \ccG)$ using Equation \eqref{eq: P help measure}. Moreover, 
let $Q'\in\mathcal M_e(\mathbb F)$ be an equivalent martingale measure such that 
\begin{align}\label{eq:equation for Q'}
    p<E_Q \Big[\sum_{i=1}^{N} \mathds{1}_{A_i} E_{P^i}[X^1\vert  \ccF] \Big]=E_{Q'} \Big[\sum_{i=1}^{N} \mathds{1}_{A_i} E_{P^i}[X^1\vert  \ccF] \Big]   
\end{align}
and such that the density of $Q'$ with respect to $P'$ is bounded, cf.\@ Lemma \ref{lem: bounded martingale measures}. Finally, we define the measure $R$ on $(\Omega,\ccG)$ as
\begin{align*}
	R(G)%
	=E_{ Q'}\Big[\sum\limits_{i=1}^N\mathds{1}_{A_i}E_{\P^i}[\mathds{1}_G\vert \mathscr F]\Big].
\end{align*} 
 Given that $ Q'\odot \P^i\vert_{\ccF}= Q'$, cf.\@ Remark \ref{rem: definitionQP}, the measure $R$ is a probability measure on $(\Omega,\mathscr G)$ and satisfies $R\vert_{\ccF}= Q'$. We find that $S$ is also a $(\mathbb F, R)$-martingale.
Moreover, we obtain for all admissible insurance strategies $\Psi$ with positive total mass $\gamma>0$ that
\allowdisplaybreaks
\begin{align}
	E_{R}[V_T(\Psi)]&=E_{R}\Big[\lim\limits_{n\to\infty}\sum\limits_{i\in\mathbb N}\psi^{i,n}(p-X^i)\Big] \nonumber \\
	&=E_{Q'}\Big[\sum\limits_{i=1}^N\mathds{1}_{A_i}E_{\P^i}\Big[\lim\limits_{n\to\infty}\sum\limits_{i\in\mathbb N}\psi^{i,n}(p-X^i)\big \vert \ccF\Big]\Big] \nonumber \\
	&=E_{P'}\Big[\Big(\frac{dQ'}{dP'\vert_{\mathscr F}}\Big)\sum\limits_{i=1}^N\mathds{1}_{A_i}E_{\P^i}\Big[\lim\limits_{n\to\infty}\sum\limits_{i\in\mathbb N}\psi^{i,n}(p-X^i)\big \vert \ccF\Big]\Big] \nonumber \\
	&=\sum\limits_{i=1}^NE_{P_i}\Big[\Big(\frac{dQ'}{dP'\vert_{\mathscr F}}\Big)\mathds{1}_{A_i}E_{\P^i}\Big[\lim\limits_{n\to\infty}\sum\limits_{i\in\mathbb N}\psi^{i,n}(p-X^i)\big \vert \ccF\Big]\Big] \nonumber \\
	&=E_{Q'}\Big[\sum\limits_{i=1}^N\mathds{1}_{A_i}\big(\gamma(p-E_{\P^i}[X^1\vert \ccF])\big)\Big]\nonumber \\
	&=\gamma\Big(p-E_{ Q'}\Big[\sum\limits_{i=1}^N\mathds{1}_{A_i}E_{\P^i}[X^1\vert \ccF]\Big]\Big )\nonumber \\
	&>0, \label{eq:ExpectationSmaller0}
\end{align}
where we use \eqref{eq:pSmaller} and \cite[Proposition B.1]{artzner2022insurance}. We now define the process $Z$ as in \eqref{eq:Zprocess}. Using \eqref{eq:ExpectationSmaller0} it follows that $Z$ is a $(\mathbb{F},R)$-martingale with $Z_0\leq 0$. Let $\xi$ be some $\mathbb{F}$-predictable strategy. Then, as in \cite[Remark 9.5]{follmer2016stochastic} the value process $(\xi \cdot S)$ is a local $(\mathbb{F}, R)$-martingale and thus $(\xi \cdot S) + Z$ is a local $(\mathbb{F}, R)$-supermartingale. For the sake of contradiction we assume that $(\xi,\Psi)$ is a $\ccP$-robust asymptotic insurance-finance arbitrage such that $\Psi$ has positive total mass. Thus, it holds that
\begin{align}\label{eq:arbitragecondition}
	(\xi\cdot S)_T+Z_T\geq 0 \quad  P\text{-a.s. for all }\P\in\ccP.
\end{align}
Given that $R\ll \tfrac{1}{N}\sum\limits_{i=1}^n P^i$, equation (\ref{eq:arbitragecondition}) is also true $R$-a.s. Thus, according to \cite[Proposition 9.6]{follmer2016stochastic} the process $(\xi \cdot S) + Z$ is a $(\mathbb{F}, R)$-supermartingale and we obtain
$$E_{R}[(\xi\cdot S)_T+Z_T]\leq (\xi \cdot S)_0 +Z_0<0.$$
This contradicts \eqref{eq:arbitragecondition} because there cannot exist an admissible insurance strategy $\Psi$ with positive total mass $\gamma>0$ which fulfills \eqref{eq:arbitragecondition}. However, given that the pure financial market is arbitrage-free with respect to $\mathbb F$-predictable trading strategies $\xi$, there could also not be an arbitrage $(\xi,\Psi)$ such that $\Psi$ has total mass $\gamma = 0$. Overall, this leads to a contradiction and thus the result is proven.
\end{proof}
\begin{remark}
Let us compare Theorem \ref{theorem:characterizationRobustInsuranceFinance} in the case of $\ccP=\{P\}$ with Theorem 1 in \cite{rasonyi2003equivalent}. According to Assumptions \ref{assum:ConditionalExpectations}  and  \ref{assum:UniformBoundedness}, each condition $(i)$ and $(ii)$ in Theorem \ref{theorem:characterizationRobustInsuranceFinance} implies the absence of arbitrage in the sense of Definition \ref{def:RobustInsuranceFinance} and there is no need for an additional  boundedness assumption on the density of the corresponding martingale measures.  By contrast, in \cite{rasonyi2003equivalent} the boundedness assumption on the density of the equivalent martingale measure is essential and cannot be substituted by means of Lemma \ref{lem: bounded martingale measures}, cf. Example 2 in \cite{rasonyi2003equivalent}.  
\end{remark}

Motivated by Theorem \ref{theorem:characterizationRobustInsuranceFinance} and the previous discussion  (see e.g.,\@ Equation \eqref{eq:ConditionPremiumNoArbitrage}), we define the following robust version of the $QP$-rule.
\begin{definition} \label{def:RobustQPRule}
Let $\ccP \subseteq \mathfrak{P}_{\mathcal{N}}(\Omega, \ccG)$ and $Q \in \cM_e(\mathbb{F})$. Then, for $X \geq 0$ $\ccP$-q.s.\@, we define the \emph{$Q  \ccP$-evaluation of $X$} by
\begin{equation} 
    E_{Q \odot \ccP}[X]:= E_{Q}\Big[\esssupN_{ P \in \ccP}E_{P}[X\vert \ccF]\Big].
\end{equation}
\end{definition}
Note that $Q \odot \ccP$ does not define a probability measure as it is the case for $Q \odot P$. 

\begin{remark}\label{rem:directedupwards}
If the set $\{E_{P}[X\vert \ccF]\vert P\in \ccP\}$ is directed upward, i.e., for all $P,P'\in  \ccP$ there exists $\tilde P\in \ccP$ such that 
\begin{align}\label{eq:directedupwards}
   \max\{ E_{P}[X\vert \ccF], E_{P'}[X\vert \ccF]\}\leq E_{\tilde P}[X\vert \ccF],
\end{align}
then, as demonstrated by \cite[Theorem A.37]{follmer2016stochastic}, there exists a sequence $(P^n)_{n\in \mathbb N}\subset \ccP$ such that 
\begin{align*}
   E_{P^n}[X\vert \ccF]\nearrow  \esssupN_{ P \in \ccP}E_{P}[X\vert \ccF]\quad P_0\text{-a.s.}\text{ for }n\to\infty.
\end{align*}
Using monotone convergence we find that
\begin{align*}
    E_{Q}\Big[\esssupN_{ P \in \ccP}E_{P}[X\vert \ccF]\Big]&=  \lim\limits_{n\to\infty}E_{Q}\Big[E_{P^n}[X\vert \ccF]\Big]\\ &\leq \sup_{ P \in \ccP}E_{Q}\left[E_{P}[X\vert \ccF]\right]\\ &\leq E_{Q}\Big[\esssupN_{ P \in \ccP}E_{P}[X\vert \ccF]\Big]
\end{align*}
and  that
\begin{align*}
    E_{Q \odot \ccP}[X]=\sup_{ P \in \ccP}E_{Q}\left[E_{P}[X\vert \ccF]\right].
\end{align*}
Consequently, for a set of priors $\ccP$ that is directed upwards in the sense of Equation \eqref{eq:directedupwards}, there is no $\ccP$-robust asymptotic insurance-finance arbitrage if and only if there is no asymptotic insurance-finance arbitrage with respect to $P$ for all $P\in \ccP$.
\end{remark}

\begin{remark}
We now briefly consider the case of $\mathbb{G}$-trading strategies on the financial market introduced in Section \ref{section:FinancialMarket}, i.e., $\xi$ is a $d$-dimensional $\mathbb{G}$-predictable process.  This reflects the fact that the insurer has access to  information on the financial market as well as on the insurance market. In order to define  a no-arbitrage condition for the financial market, we introduce some more notation. Let $\ccP \subseteq \mathfrak{P}_{\mathcal{N}}(\Omega, \ccG).$ We say that a measure $Q \in \mathfrak{P}(\Omega, \cG)$ is \emph{dominated} by $\ccP$ if there exists $\P\in \ccP$ such that $  Q\ll P$, and in this case we write $Q \lll \ccP$. Next, we define the set
\begin{equation} \label{eq:NoArbitrageMeasuresFinancialMarketG}
\cM(\mathbb{G}):= \lbrace Q \in \mathfrak{P}(\Omega, \cG) \vert \  Q \lll \ccP \text{ and } S \text{ is a } (Q, \mathbb{G})\text{-martingale}\rbrace.
\end{equation}
In this case we assume that for all $P \in \ccP$ there exists $Q \in \cM(\mathbb{G})$ such that $P \ll Q $. Then, according to \cite[Theorem 4.5]{bouchard_nutz_2015}, there is no $\ccP$-robust arbitrage, denoted by $\textnormal{NA}(\ccP,\mathbb{G})$, on the financial market, which means that for all $\mathbb{G}$-trading strategies
\begin{equation}
    (\xi \cdot S)_T \geq 0 \quad \ccP\text{-q.s.} \quad \text{ implies } \quad (\xi \cdot S)_T = 0 \quad \ccP\text{-q.s.} 
\end{equation}
Given that  every $\mathbb{F}$-trading strategy is also a $\mathbb{G}$-trading strategy, it is clear that $\textnormal{NA}(\ccP, \mathbb{G})$ implies $\textnormal{NA}(\ccP, \mathbb{F})$, where $\mathbb G$ and $\mathbb{F}$ refers here to the $\mathbb G$-trading strategies and $\mathbb{F}$-trading strategies, respectively. It thus follows that $\textnormal{NRIFA}(\ccP, \mathbb{G})$ implies $\textnormal{NRIFA}(\ccP, \mathbb{F})$ and thus (i) or (ii) in Theorem \ref{theorem:characterizationRobustInsuranceFinance} is satisfied. Here, $\textnormal{NRIFA}(\ccP, \mathbb{G})$ is defined as in Definition \ref{def:RobustInsuranceFinance}, but with a $\mathbb{G}$-trading strategy $\xi$. However, the corresponding \emph{if} direction of Theorem \ref{theorem:characterizationRobustInsuranceFinance} is more delicate and could therefore serve as a topic  for future research. %
\end{remark}

\section{Robust two-step evaluation} \label{sec:TwoStepValuation}
In this section we show that Theorem \ref{theorem:characterizationRobustInsuranceFinance} provides a theoretical foundation for the so-called \emph{two-step evaluation},  cf.\@ \cite{pelsser_stadje_2014}, \cite{dhaene2017fair}. This kind of evaluationis used for the pricing of hybrid products depending on the financial market, as well as on other random sources, e.g.,\@ individual risks depending on the policy holder of an insurance contract. The idea of a two-step evaluation is to combine actuarial techniques with concepts from financial mathematics. In the following, we recap the idea and the concept.  
However, note that in contrast to the existing literature %
we do not fix any probability measure on $(\Omega,\mathscr G),$ but only a prior $ P_0$ on the measurable space $(\Omega,\mathscr F)$.

Let $X$ be a $\ccG$-measurable random variable, representing the discounted payoff of an insurance product. In a first step we consider the $\ccF$-conditional risk of $X$, i.e.\@ $\rho_{\ccF}(X)$, where $\rho_{\ccF}$ is a suitable $\ccF$-conditional risk measure defined on the space of bounded random variables on $(\Omega, \ccG)$, which is denoted by $\mathcal{L}_b(\Omega,\mathscr G)$. This corresponds to an actuarial evaluation resulting in a $\ccF$-measurable random variable $\rho_{\ccF}(X)$ defined on the financial market $(\Omega, \ccF, P_0).$ In the second step we price $\rho_{\ccF}(X)$ on the financial market under an equivalent risk-neutral measure $Q \in \cM_e(\mathbb{F})$. Combining these, we arrive at the following two-step evaluation:
\begin{equation}\label{eq:TwoStepValuationPellser_Stadje}
    \pi\colon  \mathcal{L}_b(\Omega,\mathscr G) \to \mathbb R,\quad  \pi(X)=E_{Q}[ \rho_{\ccF}(-X)].
\end{equation}
As already mentioned in \cite{artzner2022insurance}, the $Q P$-evaluation in \eqref{eq:ConditionPremiumNoArbitrage} is a two-step evaluation with the $\ccF$-conditional risk measure $\rho_{\ccF}(X)= E_{P}[X\vert \mathscr F]$ as well as the new concept of the robust $Q\ccP$-evaluation from Definition \ref{def:RobustQPRule}.   In the following we recap some well-known facts for conditional risk measures in order to highlight that for specific coherent $\ccF$-conditional risk measures $\rho_{\ccF}$ the two-step evaluation in \eqref{eq:TwoStepValuationPellser_Stadje} can be rewritten by a  $Q  \ccP$-evaluation for a suitable choice of $\ccP$.%
\subsection{Robust representation of conditional risk measures} 
For the reader's convenience we recall the definition of conditional risk measures (see e.g., \cite{follmer2016stochastic}).
\begin{definition}\label{example:SpecialCaseDominated}
A map $\rho_{\ccF}:\mathcal{L}_b(\Omega, \ccG) \to L^{\infty}(\Omega, \ccF, P_0)$ is called a \emph{convex $\ccF$-conditional risk measure}, if for all $X,Y \in \mathcal{L}_b(\Omega, \ccG)$ the following holds $P_0$-a.s.:
\begin{enumerate}
    \item \emph{Conditional cash invariance:} $\rho_{\ccF}(X+\bar{X})= \rho_{\ccF}(X)-\bar{X}$ for any $\bar{X} \in \mathcal{L}_b(\Omega, \ccF)$.
    \item \emph{Monotonicity:} $X \leq Y$ implies $ \rho_{\ccF}(X) \geq \rho_{\ccF}(Y)$.
    \item \emph{Normalization:} $\rho_{\ccF}(0)=0$.
    \item \emph{Conditional convexity:} $ \rho_{\ccF}(\lambda X + (1-\lambda)Y) \leq \lambda \rho_{\ccF}(X)+(1-\lambda)\rho_{\ccF}(Y)$ for $\lambda \in \mathcal{L}_b(\Omega, \ccF)$ with $0 \leq \lambda \leq 1$.
\end{enumerate}
A convex $\ccF$-conditional risk measure $\rho_{\ccF}$ is called \emph{coherent} if it also satisfies the following condition:
\begin{enumerate}
    \item[(v)] \emph{Conditional positive homogeneity:} $\rho_{\ccF}(\lambda X) = \lambda \rho_{\ccF}(X)$ for $\lambda \in \mathcal{L}_b(\Omega, \ccF)$ with $\lambda \geq 0.$
\end{enumerate}
We say that  $\rho_{\ccF}$ is continuous from below if 
\begin{enumerate}
    \item[(vi)] \emph{Continuity from below:} $X_n\nearrow X$ pointwise on $\Omega$  implies $\rho_{\ccF}(X_n)\searrow \rho_{\ccF}(X)$.
\end{enumerate}
Moreover, we say that $\rho\colon \mathcal{L}_b(\Omega,\ccG)\to\mathbb R$ is a \emph{(coherent) convex risk measure} if $\mathscr F=\{\Omega,\emptyset\}$. \end{definition}
\begin{definition} Let $S$ be a financial market on $(\Omega,\mathscr F,\mathbb F, P_0)$ and denote by $\mathcal M_e(\mathbb F)$ the set of all martingale measures which are equivalent to $P_0$. A map $\pi\colon \mathcal{L}_b(\Omega,\mathscr G) \to \mathbb R$ is called \emph{two-step evaluation} if there is a $\ccF$-conditional risk measure $\rho_{\ccF}$ and an
equivalent martingale measure $ Q\in\mathcal M_e(\mathbb F)$ such that $\pi(X)=E_{Q}[\rho_{\ccF}(-X)]$ for all $X\in \mathcal{L}_b(\Omega,\mathscr G).$
\end{definition}
We now show that Theorem \ref{theorem:characterizationRobustInsuranceFinance} provides an economic foundation for the pricing of finance-linked insurance products via two-step evaluations. Indeed, by using results from \cite{follmer2016stochastic}, we formulate sufficient condition for a conditional risk measure $\rho_{\ccF}$ in a two-step evaluation $\pi(X)=E_{Q}[\rho_{\ccF}(-X)]$ such that $\pi$ is $ Q \mathscr P$\nolinebreak-evaluation. In this way, the two-step evaluation $\pi$ characterizes the $\ccP$-robust asymptotic insurance-finance arbitrage-free price as characterized in Theorem \ref{theorem:characterizationRobustInsuranceFinance}. 
\begin{lemma}\label{thm:riskmeasurerepresentation}
    Let $\rho_{\ccF}\colon\mathcal{L}_b(\Omega,\mathscr G)\to L^{\infty}(\Omega,\ccF,P_0)$ be a convex $\mathscr F$-conditional risk measure which is continuous from below. Then $\rho_{\ccF}$ is represented by
    \begin{align}\label{eq:representationconvexriskmeasure}
        \rho_{\ccF}(X)=\esssupZ_{P \in \mathfrak{P}^{P_0}(\Omega,\mathscr G)} \left( E_{P}[-X \vert \ccF]-\alpha^{\textnormal{min}}_{\ccF}(P)\right),
    \end{align}
    where the acceptance set $\mathscr A_{\ccF}$, the penalty function $\alpha^{\textnormal{min}}_\mathscr{F}$ and the set of priors $\mathfrak{P}^{\mathbb P}(\Omega,\mathscr G)$ are defined by
    \begin{align*}
        \mathscr A_{\ccF}&:=\left \lbrace X\in\mathcal{L}_b(\Omega,\mathscr G) \ \vert \ \rho_{\ccF}(X)\leq 0\right \rbrace,
    \end{align*}
    \begin{align*}
        \alpha^{\textnormal{min}}_\mathscr{F}(P)&:=\esssupZ_{X\in \mathscr A_{\ccF}}E_{P}[-X \vert \ccF],
    \end{align*}
    and
\begin{align} \label{eq:SetPriorsEqualOnMarektFiltration}
\mathfrak{P}^{P_0}(\Omega,\mathscr G):= \lbrace P \in \mathfrak{P}(\Omega,\cG) \  \vert \ P \vert_{\cF} =  P_0   \rbrace\subseteq \mathfrak{P}_{\mathcal N}(\Omega,\mathscr G). 
\end{align}
If, in addition, $\rho_{\ccF}$ is a coherent $\ccF$-conditional risk measure, then there is a subset $\mathscr P\subseteq \mathfrak{P}^{P_0}(\Omega,\mathscr G)$ such that
    \begin{align}\label{eq:representationcoherentriskmeasure}
        \rho_{\ccF}(X)=\esssupZ_{P \in  \ccP} E_{P}[-X \vert \ccF].
    \end{align}
\end{lemma}
\begin{proof}
In the unconditional case the first statement follows from Theorem 4.16 and Theorem 4.22 as laid out by \cite{follmer2016stochastic}. Using the same idea as in the proof of Theorem 11.2 in \cite{follmer2016stochastic} yields the conditional statement. 
For the second statement we refer to Corollary 4.19 for the unconditional case. Using similar arguments to those put forth in Corollary 11.6 in \cite{follmer2016stochastic}, the conditional statement follows.
\end{proof}

Thus, Lemma \ref{thm:riskmeasurerepresentation} shows that every two-step evaluation $\pi(X)=E_{Q}[\rho_{\ccF}(-X)]$ given by an equivalent martingale measure $ Q\in\mathcal M_e(\mathbb F)$ and a coherent $\ccF$-conditional risk measure $\rho_{\ccF}$ which is continuous from below can be written as $ Q \mathscr P$\nolinebreak-evaluation for a suitable subset $\ccP \subseteq \mathfrak{P}^{P_0}(\Omega, \ccG)$.

\begin{remark}
If we a priori fix the nullsets on $(\Omega,\mathscr G)$ by a probability measure $P$ on $(\Omega,\mathscr G)$ such that $P\vert_{\mathscr F}= P_0$, then we can define the conditional risk measure $\rho_{\ccF}$ on $L^{\infty}(\Omega,\mathscr G,P)$, instead of working with $\mathcal{L}_b(\Omega, \ccG).$  In this case $\rho_{\ccF}$ only needs to be continuous from above, as opposed to satisfying the stronger assumption of continuity from below, in order to have a representation as in Lemma \ref{thm:riskmeasurerepresentation}, cf. Theorem 4.33 and Theorem 11.2 in \cite{follmer2016stochastic}. The drawback of this approach is that we consider uncertainty in a narrow sense because we fix all relevant nullsets on $(\Omega, \ccG)$ using a single probability measure $P$. Nevertheless, this also leads to a robust pricing problem in the spirit of Section \ref{sec:RobustAsymptoticInsuranceFinanceArbitrage}, which will be discussed in more detail in Remark \ref{remark:IIDCopiesProblems}.
\end{remark}
Next, we provide some examples for the set $\ccP \subseteq \mathfrak{P}^{P_0}(\Omega, \ccG)$ in Lemma \ref{thm:riskmeasurerepresentation} and the associated $\ccF$-conditional risk measure in Equation \eqref{eq:representationcoherentriskmeasure}.
\begin{example}\label{example:conditionalaverageatrisk}
Let $P$ be a probability measure on $(\Omega,\mathscr G)$ such that $P\vert_{\ccF}= P_0$. We consider the set of priors $\ccP_{\lambda} \subseteq \mathfrak{P}^{ P_0}(\Omega, \ccG)$ given by
\begin{align}
    \ccP_{\lambda}\coloneqq \Big \lbrace \tilde{P} \in \mathfrak{P}^{ P_0}(\Omega, \ccG)\ \vert\  \tilde{P} \ll {P}\text{ with } d\tilde{P}/d{P} \leq\lambda^{-1} \ {P}\text{-a.s.}  \Big \rbrace \quad \text{ for } \lambda \in (0,1).
\end{align}
In this case, the associated risk measure $\rho_{\ccF}(X)$ is the \emph{conditional Average value at risk}, denoted by $\textnormal{AV@R}_{\lambda}(X \vert \ccF)$, see also Definition 11.8 in \cite{follmer2016stochastic}. Note that in this case the set $\ccP_{\lambda}$ is dominated by the probability measure $P$. 
\end{example}

\begin{example}\label{example:coherentEntropicRiskMeasure}
Let $P$ be a probability measure on $(\Omega,\mathscr G)$ such that $P\vert_{\ccF}= P_0$. We consider the set of priors $\ccP_{c} \subseteq \mathfrak{P}^{P_0}(\Omega, \ccG)$ for $c > 0$ given by
\begin{align}
    \ccP_{\!c}\coloneqq \Big \lbrace \tilde{P} \in \mathfrak{P}^{P_0}(\Omega, \ccG)\ \vert\  H(\tilde{P} \vert  P) \leq c \Big \rbrace,
\end{align}
where $H(\tilde{P} \vert P)$ denotes the \emph{relative entropy of $\tilde{P}$ with respect to $P$} and is defined by
\[
H(\tilde{P} \vert P ):=\left\{%
\begin{array}{ll}
    E_{\tilde{P}}\Big[ \log \frac{d \tilde{P}}{dP}\Big], & \hbox{if } \tilde{P} \ll P \\
    + \infty, & \hbox{otherwise.} \\
\end{array}%
\right.
\]
Here, the associated risk measure $\rho_{\ccF}(X)$ is the \emph{coherent entropic risk measure}, introduced in \cite{follmer2011entropic}. As in Example \ref{example:conditionalaverageatrisk}, the set $\ccP_{\!c}$ is dominated by the probability measure \nolinebreak $P$. 
\end{example}
In the following Proposition we show that pricing with  two-step evaluations leads to arbitrage-free premiums in the sense of Section \ref{sec:RobustAsymptoticInsuranceFinanceArbitrage}. This remarkable result is a consequence of Lemma \ref{thm:riskmeasurerepresentation} and Theorem \ref{theorem:characterizationRobustInsuranceFinance}.
\begin{proposition}
Let $S$ be a financial market on $(\Omega,\mathscr F)$ and let $\pi$ be a two-step evaluation with $\mathscr F$-conditional convex risk measure of the form \eqref{eq:representationcoherentriskmeasure} and set $\mathscr P=\{P\in \ccP_{\mathcal N}(\Omega,\mathscr G)\vert \alpha^{\min}_{\mathscr F}(P)<\infty \}$. Assume that $\mathcal X=(X^i)_{i\in\mathbb N}$ is a sequence of insurance benefits  fulfilling Assumption \ref{assum:ConditionalExpectations} and assume that $p<\pi(X^1)$.  Then there is \textnormal{NRIFA}($\ccP$) with respect to the insurance-finance market $(S,\mathcal X,p)$. 
\end{proposition}
\begin{proof}
The result is a direct consequence of Theorem \ref{theorem:characterizationRobustInsuranceFinance}, since the $Q\ccP$-evaluation is an upper bound for the two-step evaluation $\pi$ and the chosen premium $p$ is even smaller by assumption.
\end{proof}

\subsection{Construction of conditional iid copies} \label{sec:ConstructionIID}
Our next goal is to apply Theorem \ref{theorem:characterizationRobustInsuranceFinance} in the context of a robust two-step evaluation. More specifically, given a random variable $\tilde{X}$ describing an insurance benefit, we determine a robust arbitrage-free premium $p$ for $\tilde{X}$ by using Theorem \ref{theorem:characterizationRobustInsuranceFinance} and by taking into account some actuarial constraints, which are reflected by the set of priors $\ccP \subseteq \mathfrak{P}_{\mathcal N}(\Omega, \ccG)$ (see e.g.,\@ the sets $\ccP_{\lambda}$ and $\ccP_{c}$ in Example \ref{example:conditionalaverageatrisk} and \ref{example:coherentEntropicRiskMeasure}, respectively).

To do so, two factors must to be considered. First, the assumptions of Theorem \ref{theorem:characterizationRobustInsuranceFinance} must be satisfied. To this end, we construct a sequence of benefits $(X^j)_{j \in \mathbb{N}}$ which are copies of $\tilde{X}$ and which satisfy Assumption \ref{assum:ConditionalExpectations}. Second, the set of priors $\ccP$ must be shifted to the product space where we model the benefits $(X^j)_{j \in \mathbb{N}}$. We observe that these steps contain some subtleties which we discuss in more detail in Remark \ref{remark:IIDCopiesProblems} after formally introducing the setting. \\

Let $(\Omega^F,\mathscr F^F)$ and $(\Omega^I,\mathscr F^I)$ be two measurable spaces on which we model purely financial and purely insurance events, respectively. On the product space $(\Omega^F\times\Omega^I,\mathscr F^F\otimes \mathscr F^I)$ we introduce the stochastic process $\tilde S=(\tilde S_t)_{t=0,..., T}$ and the random variable $\tilde X$ describing the financial market and a single insurance benefit, respectively. Moreover, let $ P_0$ be a measure on $(\Omega^F\times \Omega^I,\mathscr F^{F}\otimes \{\emptyset,\Omega^I\})$ which determines the nullsets in $\mathscr F^F\otimes\{\emptyset,\Omega\}$ and $\ccP$ be a set of probability measures on $(\Omega^F\times\Omega^I,\mathscr F^F\otimes \mathscr F^I)$ such that 
$$P\vert_{\mathscr F^{F}\otimes \{\emptyset,\Omega^I\}}\sim P_0 \text{ for all } P\in \ccP.$$ 
We now shift $\ccP$ to a set of priors $\mu_{\ccP}$ on $(\Omega^F\times(\Omega^I)^{\mathbb N},\mathscr F^F\otimes (\mathscr F^I)^{\otimes \mathbb N})$. Furthermore, on this space we copy the financial market $\tilde S$ to $S$ and construct insurance benefits $(X^j)_{j\in\mathbb N}$ which are \emph{iid} conditionally on $S$ such that the law of $(S,X^j)$ for $j \in \mathbb{N}$ under every $\mu^P \in \mu_{\ccP}$ coincides with the law of $(\tilde{S},\tilde{X})$ under $P \in \ccP$. %
\begin{remark} \label{remark:IIDCopiesProblems}
We also emphasize that if the set $\ccP$ is dominated by a measure $P \in \ccP$, as is the case in  Example \ref{example:conditionalaverageatrisk} and \ref{example:coherentEntropicRiskMeasure}, this will no longer hold for the shifted set $\mu_{\ccP}$. The reason for this is that absolute continuity of measures is not stable under countable products. Therefore, the seemingly not robust problem in the dominated case on $(\Omega^F\times\Omega^I,\mathscr F^F\otimes \mathscr F^I,\ccP)$ is indeed a robust pricing problem on $(\Omega^F\times(\Omega^I)^{\mathbb N},\mathscr F^F\otimes (\mathscr F^I)^{\otimes \mathbb N},\mu_{\ccP})$. 
\end{remark}
To be precise, we define 
\begin{align*}
   \tilde\Omega &\coloneq \Omega^F\times \Omega^I, &\Omega &\coloneq \Omega^F\times (\Omega^I)^{\mathbb N},\\
   \tilde{\mathscr F} &\coloneqq \mathscr F^F\otimes \{\Omega^I,\emptyset\}, &\mathscr F &\coloneqq \mathscr F^F\otimes \{\Omega^I,\emptyset\}^{\otimes \mathbb N},\\
    \tilde{\mathscr G} &\coloneqq \mathscr F^F\otimes \mathscr F^I, &\mathscr G &\coloneqq \mathscr F^F\otimes (\mathscr F^I)^{\otimes \mathbb N}.
\end{align*}
We denote by $\tilde\omega=(\tilde\omega^F,\tilde\omega^I)$ an element in $\tilde \Omega$ and by $\omega=(\omega^F,(\omega^I_j)_{j\in\mathbb N})$ an element in $\Omega$. Furthermore, we introduce the following projections on $\tilde{\Omega}$:
\begin{align*}
   \tilde\pi_{\Omega^F}&\colon \tilde \Omega  \to \Omega^F\!\!, \quad\tilde{\pi}_{\Omega^F}(\tilde\omega)=\tilde\omega^F\!\!,\\
    \tilde\pi_{\Omega^I}&\colon \tilde \Omega\to \Omega^I,\quad\tilde{\pi}_{\Omega^I_j}(\tilde\omega)=\tilde\omega^I,
\end{align*}
as well as the following the projections on $\Omega$:
\begin{align*}
    \pi_{\Omega^F}&\colon \Omega \to \Omega^F\!\!, \quad\pi_{\Omega^F}(\omega)=\omega^F\!\!,\\
    \pi_{\Omega^I_j}&\colon \Omega \to \Omega^I,\quad\pi_{\Omega^I_j}(\omega)=\omega^I_j.
\end{align*}
Given a measure $P$ on $(\tilde\Omega,\tilde{\mathscr G})$, the aim is to define a probability measure $\mu_P$ on $(\Omega,\mathscr G)$ which fulfills the following properties:
\begin{align}\label{eq:mupcondition1}
    \text{The law of }(\pi_{\Omega^F},\pi_{\Omega^I_j}) \text{ under }\mu_P \text{ equals } P \text{ for all }j\in\mathbb N,
\end{align}
and 
\begin{align}\label{eq:mupcondition2}
    (\pi_{\Omega^I_j})_{j\in\mathbb N}\text{ are  $\mathscr F$-conditionally independent under }\mu_P.
\end{align}
If $P$ is a product measure given by $P=P_F\otimes P_I$ for measures $P_F$ on $(\Omega^F,\mathscr F^F)$ and $P_I$ on $(\Omega^I,\mathscr F^I)$, then the measure $\mu_P$ can be defined by $\mu_P=P_F\otimes (P_I)^{\otimes \mathbb N}$. Otherwise, we construct $\mu_P$ via disintegration as follows. For some measure $P$ on $(\tilde\Omega,\tilde{\mathscr G})$ the measure $\mu_{p}$ is defined as 
\begin{align}\label{eq:disintegration}
    \mu_P(A\times B)\coloneqq\int \mathds{1}_{\tilde\pi^{-1}_{\Omega^F}(A)}(P^{\tilde\pi_{\Omega^I}\vert \tilde{\mathscr F}})^{\otimes\mathbb N}(B)dP \quad\text{for }A\in \mathscr F^F \text{ and }B\in  (\mathscr F^I)^{\otimes \mathbb N}.
\end{align}
where $P^{\tilde\pi_{\Omega^I}\vert \tilde{\mathscr F}}$ denotes the regular version of the conditional probability of $\tilde\pi_{\Omega^I}$ given $\tilde{\mathscr F}$ (see e.g.,\@ \cite[Chapter 8]{kallenberg1997foundations}). Note that we implicitly assume its existence. This is no restriction, however, because $\Omega^F$, representing a financial market with $d+1$ assets and $T$ time steps, can always be assumed to have the form $\Omega^F=\mathbb R^{(d+1)\times(T+1)}$ and, thus, is a Borel space. Moreover, using a monotone class argument, it follows that
\begin{align*}
    (P^{\tilde\pi_{\Omega^I}\vert \tilde{\mathscr F}})^{\otimes\mathbb N}\colon &(\mathscr F^I)^{\otimes\mathbb N}\times \tilde \Omega\to [0,1]\\
    &(B,\tilde\omega)\mapsto (P^{\tilde\pi_{\Omega^I}\vert \tilde{\mathscr F}}(\cdot,\tilde\omega))^{\otimes\mathbb N}(B)
\end{align*}
is a probability kernel from $(\tilde{\Omega},\tilde{\mathscr G})$ to $(\Omega,\mathscr G)$ and, thus, the measure $\mu_P$ is well defined. If $P$ is a measure on $(\tilde\Omega, \tilde{\mathscr F})$, then \eqref{eq:disintegration} defines a measure on $(\Omega,\mathscr F)$. In this case we have  $B=\Omega^{\mathbb N}$.  Note that by using \eqref{eq:disintegration} we get the following: 
\begin{align*}
    \mu_{P}(A\times (\Omega^I)^{\mathbb N})= P(A\times \Omega^I)\quad\text{ for all }A\in \mathscr F^F
\end{align*}
and thus $P\vert_{\tilde {\mathscr F}}\sim P_0 $ implies $\mu_P\vert_{\mathscr F}\sim\mu_{P_0}$ for all $P\in \ccP$.

Let $\tilde{\mathbb F}=(\tilde{\mathscr F_t})_{t\leq T}$ with $\tilde{\mathscr F_t}\coloneqq \mathscr F_t^F\vee \{\Omega^I,\emptyset\}$, be a filtration on $(\tilde{\Omega},\tilde{\mathscr F})$. Hereafter, we assume that $ \tilde S=(\tilde S_t)_{t\leq T}$ is a $\tilde{\mathbb F}$-adapted stochastic process on $(\tilde{\Omega},\tilde{\mathscr F})$ describing the prices in a financial market. Moreover, let $\mathcal M_{e}(\tilde{\mathbb F})$ be the set of all martingale measures on $(\tilde{\Omega}, \tilde{\mathcal{F}})$ which are equivalent some fixed measure $P_0$ on $(\tilde\Omega,\tilde{\mathscr F})$. The insurance and financial filtration on $(\tilde{\Omega},\tilde{\mathscr G})$ is denoted by $\tilde{\mathbb G}=(\tilde{\mathcal{G}}_t)_{t \leq T}$ and the insurance benefit, a random variable on $(\tilde{\Omega},\tilde{\mathscr G})$, is denoted by $\tilde X$. 
In order to shift all quantities to $(\Omega,\mathscr G)$ we define
\begin{align*}
\mathbb F &=(\mathscr F_t)_{t\leq T} \quad \text{with  } \quad \mathscr F_t\coloneqq \mathscr F_t^F\vee \{\Omega^I,\emptyset\}^{\otimes\mathbb N},\\
\mathbb G&=(\mathscr G_t)_{t\leq T} \quad  \text{ with }  \quad \mathscr G_t\coloneqq \sigma((\pi_{\Omega^F},\pi_{\Omega^I_j})  \colon   \Omega\to (\tilde \Omega,\tilde{\mathscr G}_t)  \vert  \ j\in\mathbb N ),
\end{align*}
and
\begin{align*}
    S_t \coloneqq  \tilde{S_t}\circ(\pi_{\Omega^F},\pi_{\Omega^I_1})&=\tilde{S_t}\circ(\pi_{\Omega^F},\pi_{\Omega^I_j}) \quad \text{ for all } j\in\mathbb N, \,t=0,...,T,\\
    X^j&\coloneqq  \tilde X\circ (\pi_{\Omega^F},\pi_{\Omega^I_j}) \quad \text{ for all }j\in \mathbb N. 
\end{align*}
Note that $\tilde{S_t}$ is assumed to be measurable with respect to  $\tilde{\mathscr F}_t\subseteq \tilde{\mathscr F}$ and, thus, it does not depend on the second coordinate $\tilde{\omega}^I.$ 

We show that the measure $\mu_P$ defined in \eqref{eq:disintegration} satisfies the desired properties in \eqref{eq:mupcondition1} and \eqref{eq:mupcondition2} and  is  the unique measure with this property. For the reader's convenience we provide the proofs of these results in detail in the Appendix.
\begin{proposition}
The measure $\mu_P$, as defined by \eqref{eq:disintegration}, is the unique measure on $(\Omega,\mathscr G)$ which fulfills \eqref{eq:mupcondition1} and \eqref{eq:mupcondition2}.
\end{proposition}

Next, we characterize the set of all equivalent martingales measures on $(\Omega, \ccF).$
\begin{proposition}\label{prop:martingalemeasuretransformation}
The set $\mathcal M_e(\mathbb F)$ of all measures on $(\Omega,\mathscr F)$ such that $S$ is a $\mathbb F$-martingale and which are equivalent to $\mu_{P_0}$ 
is given by
\begin{align*}
    \mathcal M_e(\mathbb F)=\{\mu_{ Q}\vert  Q\in\mathcal M_e(\tilde{\mathbb F})\}.
\end{align*}
\end{proposition}
\begin{proposition}\label{prop:premiumchapter3}
For any $P \in \ccP$ and $ Q\in \mathcal M_e(\tilde{\mathbb F})$ we have the following:
\begin{align*}
E_{ Q}[E_{P}[\tilde X\vert \tilde{\mathscr F}] ]&=E_{\mu_{ Q}}[E_{\mu_P}[X^1\vert\mathscr F] ],\\
E_{ Q}[\esssupZ_{P \in \ccP} E_{P}[\tilde X\vert \tilde{\mathscr F}] ]&=E_{\mu_{ Q}}\left[\esssupZ_{P \in \ccP}E_{\mu_P}[X^1\vert\mathscr F] \right].
\end{align*}
\end{proposition}

We emphasize that Theorem \ref{theorem:characterizationRobustInsuranceFinance} and Proposition \ref{prop:premiumchapter3} build a foundation of two-step evaluations from a new perspective. We shift the insurance benefit to an insurance-finance market such that the assumptions for Theorem \ref{theorem:characterizationRobustInsuranceFinance} are fulfilled and characterize the robust insurance-finance arbitrage-free prices therein. Then, Proposition \ref{prop:premiumchapter3} shows that the prices in the shifted insurance-finance market coincide with the two-step evaluation of the initial benefit.  

\section{Modeling of insurance-finance markets}\label{sec:GeneralizedSetting}
In this section we provide a model for an insurance-finance market and calculate the robust insurance-finance arbitrage-free premium by means of the $Q\ccP$-evaluation, cf.\@ Theorem \ref{theorem:characterizationRobustInsuranceFinance} and Definition \ref{def:RobustQPRule}.

As in Section \ref{section:FinancialMarket}, let $S^0\equiv 1$ be the bank account and denote by $S^1=(S_t^1)_{t=0,...,T}$ the discounted price process of a risky asset on $(\Omega, \ccF).$ We fix the $\ccF$-nullsets $\mathcal{N}$ generated by a probability measure $P_0 \in \mathfrak{P}(\Omega, \ccF)$. We assume that the filtration  $\mathbb F$ is generated by $S$.
Next, we introduce the $\mathbb N_0$-valued random variables $\tau^1$ and $\tau^2$ representing the time of death and the time of surrender of a policy holder, respectively. Let the $\sigma$-algebra $\mathscr G$ be given by
$\mathscr G=\mathscr F \vee \sigma(\tau_1)\vee\sigma(\tau_2)$
and the filtration $\mathbb G$ given by $\mathbb G=\mathbb F\vee\mathbb H$, where $\mathbb{H}$ is the filtration generated by the processes $(\textbf{1}_{\lbrace \tau^1 \leq t \rbrace})_{t=0,\dots,T}$ and $(\textbf{1}_{\lbrace \tau^2 \leq t \rbrace})_{t=0,\dots,T}$. Note that  $\tau^1$ and $\tau^2$ are $\mathbb G$-stopping times but, in general, they are not $\mathbb F$-stopping times. 

Given a parameter set $\Theta$, we introduce the law of the stopping times $(\tau^1,\tau^2)$ under the set of priors $\ccP_{\Theta}=(P^{\theta})_{\theta\in \Theta}\subseteq \mathfrak{P}_{\mathcal{N}}(\Omega, \ccG)$. In particular, we assume that for each $P^{\theta} \in \ccP_{\Theta}$ the conditional laws of $\tau^1$ and $\tau^2$ are given by
\begin{align} \label{eq:ConditionalProbability}
    P^{\theta} \left[ \tau^1 \leq t \vert \ccF \right] \coloneqq  F^{\ccF}_1(\theta, t) \quad \text{ and } \quad
    P^{\theta} \left[ \tau^2 \leq t \vert \ccF \right] \coloneqq F^{\ccF}_2(\theta, t) \quad \text{ for } t \in \mathbb{N}_0,
\end{align}
where we assume that for fixed $\theta\in \Theta$ the mappings $F^{\ccF}_1(\theta,\cdot)$ and $F^{\ccF}_2(\theta,\cdot)$ are $\mathscr F$-conditional distribution functions. 

\subsection{Modeling under conditional independence}\label{sec:conditionallyindependence}

We assume that under every $P^{\theta} \in \ccP_{\Theta}$ the random variables  $\tau^1$ and $\tau^2$ are $\ccF$-conditionally independent, i.e.,
\begin{align}\label{eq:conditionalindependenceoftau}
   P^{\theta} \left[ \tau^1 \leq s, \tau^2 \leq t \vert \ccF \right]\coloneqq F^{\ccF}_1(\theta, s)F^{\ccF}_2(\theta, t) \quad\text{for } s,t\in\mathbb N_0.
\end{align}
In order to determine the law of $(S, \tau^1,\tau^2)$ under $P^{\theta}$ it is now sufficient to fix a model for the restricted measures $P^{\theta}\vert_{\mathscr F}\sim P_0$ and use disintegration. Here, we could assume that $P^{\theta}\vert_{\mathscr F}=P_0$ for all $\theta\in\Theta$ because the $Q\ccP$-evaluation is invariant under the specific choice of the measures $\{P^{\theta}\vert_{\mathscr F}\vert \theta\in \Theta \}$ as the set of equivalent martingale measures $\cM_e(\mathbb F)$ only depends on the nullsets $\mathcal N$ generated by $P_0$.

We introduce the discounted survival benefit $X_{\text{survival}}$ and the discounted surrender benefit $X_{\text{surrender}}$ using 
\raggedbottom
\begin{align} \label{eq:x_survival}
X_{\text{survival}}&\coloneqq \textbf{1}_{\lbrace \tau^1 >T, \tau^2 >T\rbrace} Y^1 (S^0_T)^{-1}\\\label{eq:x_surrender}
X_{\text{surrender}}&\coloneqq \sum_{t=1}^{T-1} \textbf{1}_{\lbrace \tau^1>t,  \tau^2=t \rbrace}Y^2_t (S^0_t)^{-1},
\end{align}
where $Y^1$ is a  $\ccF$-measurable random variable and $Y^2:=(Y_t^2)_{t=0,...,T}$ is a $\mathbb{F}$-adapted process.
The insurance benefit $X$ is then given by
\begin{align}\label{eq:benefitX} 
X&\coloneqq  X_{\text{survival}}+ X_{\text{surrender}}.
\end{align}
An insurance seeker with such a policy receives the payment $Y^1$ at the maturity $T$ if he survives until $T$ and does not surrender before time $T$. If he surrenders at time $t<T$ he receives the payment $Y^2_t$.

\begin{example}\label{example:payoffspecification}
We define the process $V=(V_t)_{t=0,...,T}$ by 
\begin{align} 
\label{eq:DefinitionPayoffGeneralizedExample}
V_t:=K(t)+(S_t-K(t))^+ \quad \text{ with } \quad K(t):=(1+r_G)^t K
\end{align}
for $K \in \mathbb{R}_+$ and $r_G>-1$. Here, $r_G$ denotes the interest rate associated with a guarantee $K$. Furthermore, we assume that in case a policy holder surrenders the contract at time $t$, he will receive the payment $(1-l)V_t$ for $l \in [0,1]$, where $l$ denotes the penalty in form of a proportional deduction of the actual value. Thus, we set $Y^1\coloneqq V_T$ and $Y_t^2\coloneqq (1-l)V_t$ for $t=0,\dots,T-1$. 
\end{example}

\begin{example}\label{example:survival_example}
We consider the parameter set $\Theta:=B \times C \coloneqq [\underline{b},\overline{b}] \times [\underline{c},\overline{c}]  \subseteq  \mathbb{R}^2_{>0}$ and assume that for every $\theta=(b,c) \in \Theta$ the conditional distribution function of the time of death in \eqref{eq:ConditionalProbability} is given by the well-known Gompertz model, i.e.,
\begin{align} \label{eq:ConditionalDistributionDeath}
    F_1^{\ccF}(\theta,t)=P^{\theta} \left[ \tau^1 \leq t \vert \ccF \right]:=1-\exp \Big(-\sum_{s=0}^{t-1}b e^{cs}\Big).
\end{align}
In this example we assume that a policy holder is not allowed to surrender, which can be realized by setting $\tau_2:=+\infty$.

We set $S_t^0:=(1+r)^t$ for some risk-free rate $r>0$ and assume that $S^1$ follows a Cox-Ross-Rubinstein model (CRR model) (see Section 5.5 in \cite{follmer2016stochastic}). This means that the stock price at time $t=1,...,T$ is given by the higher value $S_t^1=S_{t-1}^1(1+u)$ for $t=1,...,T$ with probability $0<p<1$ and by the lower value $S_t^1=S_{t-1}^1(1+v)$ for $t=1,...,T$ with probability $1-p$, such that $-1 < v<u$ and $v<r<u$. Furthermore, we denote by $\mathbb{F}= (\ccF_t)_{t=0,...,T}$ the filtration on $(\Omega, \ccF)$ given by $\ccF_t:=\sigma(S^1_0,...,S^1_t)$ for $t=0,...,T$. Let $R_t:=S_t/S_{t-1}$ for $t=1,\dots,T$.  The unique equivalent martingale measure can be characterized by the measure $Q$ on $(\Omega,\mathscr F)$ such that $R_1,\dots,R_T$ are independent and 
\begin{align} \label{eq:CRRUniqueMartingaleMeasure}
    Q[R_t=1+u]=\frac{r-v}{u-v} \quad \text{ and } \quad Q[R_t=1+v]=\frac{u-r}{u-v}\quad\text{ for all }t=1,\dots,T. 
\end{align}
We apply the $Q \ccP_{\Theta}$-evaluation on $X=X_{\text{survival}}$ given by $\eqref{eq:x_survival}$ and Example \ref{example:payoffspecification} and get
\allowdisplaybreaks
\begin{align*}
    E_{Q \odot \ccP_{\Theta}}\left[ X\right] &= E_{Q} \Big[ \esssupN_{P^{\theta} \in \ccP_{\Theta}} E_{P^{\theta }}\left[\textbf{1}_{\lbrace \tau^1 >T\rbrace} V_T (1+r)^{-T} \vert \ccF \right] \Big]\\
    &= E_{Q} \Big[  V_T (1+r)^{-T} \esssupN_{\theta \in \Theta} (1-F_1^{\ccF}(\theta,T))  \Big] \\
     &=E_{Q} \Big[V_T (1+r)^{-T} \esssupN_{\theta \in \Theta} \exp\Big(-\sum_{s=0}^{T-1} be^{cs}\Big)  \Big]\\
    &=(1+r)^{-T} \exp\Big(-\sum_{s=0}^{T-1} \underline{b}e^{\underline{c}s}\Big) \left( K(T) + \mathbb{E}_{Q} \left[(S_T - K(T))^+ \right] \right),
\end{align*}
where the arbitrage-free price for $(S_T-K(T))^+$ is given by
\begin{align*}
   E_{Q}[(S_T-K(T))^+]= \sum_{k=0}^T \binom{T}{k} \left(\frac{r-v}{u-v}\right)^k \left( \frac{u-r}{u-v}\right)^{T-k} \!\!\! (S_0^1 (1+u)^k (1+v)^{T-k} - K(T))^+.
\end{align*}
In this example the map
$   (b,c)\mapsto E_{P^{\theta}}\left[ \textbf{1}_{\lbrace \tau^1 >T\rbrace} V_T (1+r)^{-T}\mathscr F\right] 
$
is strictly decreasing in $b$ and $c$. Thus, based on Remark \ref{rem:directedupwards}, it holds that
\begin{align}\label{eq:supequation}
    E_{Q \odot \ccP_{\Theta}}\left[ \textbf{1}_{\lbrace \tau^1 >T\rbrace} V_T (1+r)^{-T}\right]=\sup_{\theta\in \Theta}E_{Q \odot P^{\theta}}\left[ \textbf{1}_{\lbrace \tau^1 >T\rbrace} V_T (1+r)^{-T}\right].
\end{align}
In other words, the $\ccP_{\Theta}$-robust price equals the worst-case price of all possible models.
\end{example}
\begin{example}
We extend Example \ref{example:survival_example} and consider an insurance benefit which includes a surrender option. The parameter set $\Theta$ is now given by   $\Theta:=A \times B \times C \times D:=[\underline{a},\overline{a}] \times [\underline{b},\overline{b}] \times [\underline{c},\overline{c}] \times [\underline{d},\overline{d}] \subseteq \mathbb{R} \times \mathbb{R}^3_{>0}$ and for every $\theta=(a,b,c,d) \in \Theta$ the conditional distribution functions of $\tau^1$ and $\tau^2$ in \eqref{eq:ConditionalProbability} are given by \eqref{eq:ConditionalDistributionDeath}
and 
\raggedbottom
\begin{align} \label{eq:ConditionalDistributionSurrender}
    F_2^{\ccF}(\theta,t)=P^{\theta} \left[ \tau^2 \leq t \vert \ccF \right]:=1-\exp\Big(-\frac{1}{d} \sum_{s=0}^{t-1}(a-S_s)^2\Big).
\end{align}

The conditional probability to surrender before or at time $t$ in \eqref{eq:ConditionalDistributionSurrender} tends to one if $(a-S_s)^2$ increases. The intuition behind this, which is related to the definition of the insurance benefit $V$ in \eqref{eq:DefinitionPayoffGeneralizedExample}, is as  follows: if the value of the asset decreases, the value of the benefit $V$ will also decline. Therefore, the insurance seeker faces the risk of ending up with $K(t)$ and thus the probability of needing to surrender increases. Conversely, if the value of the asset increases, the value of $V$ increases as well, giving the insurance seeker  an incentive to surrender. Obviously, in both cases these considerations  further depend  on the penalty parameter $l$. In summary, it is more likely that the insurance seeker surrenders in case the value of the asset deviates too much from the level $a$.

In applications, we can choose, for example,  the parameter set $\Theta$  as the confidence intervals around the empirically observed values of $(a,b,c,d).$
The $Q\ccP_{\Theta}$-evaluation of the benefit $X$ given by \eqref{eq:benefitX} and Example \ref{example:payoffspecification} is then given by
 \begin{align}
    &E_{Q \odot \ccP_{\Theta}}[X] \nonumber \\
    &=E_{Q \odot \ccP_{\Theta}}\Big[ \textbf{1}_{\lbrace \tau^1 >T, \tau^2 >T\rbrace} V_T (1+r)^{-T} +\sum_{t=1}^{T-1} \textbf{1}_{\lbrace \tau^2 =t\rbrace}\textbf{1}_{\lbrace \tau^1>t\rbrace} (1-l)V_t(1+r)^{-t}\Big]  \nonumber \\
    &= E_{Q} \Big[ \esssupN_{P^{\theta} \in \ccP_{\Theta}} E_{P^{\theta }}\Big[\textbf{1}_{\lbrace \tau^1 >T, \tau^2 >T\rbrace} V_T (1+r)^{-T} +\sum_{t=1}^{T-1} \textbf{1}_{\lbrace \tau^2 =t\rbrace}\textbf{1}_{\lbrace \tau^1>t\rbrace} (1-l)V_t(1+r)^{-t}\big \vert \ccF\Big] \Big] \nonumber\\
    &= E_{Q} \Big[ \esssupN_{\theta \in \Theta} \Big(  (1-F_1^{\ccF}(\theta,T)) (1-F_2^{\ccF}(\theta,T))V_T (1+r)^{-T}\nonumber \\
    &\quad \quad + \sum_{t=1}^{T-1}  ( F_2^{\ccF}(\theta,t)-F_2^{\ccF}(\theta,t-1)) (1-F_1^{\ccF}(\theta,t))(1-l)V_t (1+r)^{-t}\Big) \Big]  \nonumber \\
     &= E_{Q} \Big[ \esssupN_{\theta \in \Theta} \Big( \exp\big(-\sum_{s=0}^{T-1} be^{cs}\big) \exp\big(-\frac{1}{d}\sum_{s=0}^{T-1} (a-S_s )^2\big)\frac{V_T}{(1+r)^{T}} \nonumber \\
    &\quad \quad + (1-l)\sum_{t=1}^{T-1} \Big( \exp\big(-\frac{1}{d}\sum_{s=0}^{t-2} (a-S_s )^2\big)-\exp\big(-\frac{1}{d}\sum_{s=0}^{t-1} (a-S_s )^2\big) \Big) \exp\big(-\sum_{s=0}^{t-1} be^{cs}\big)\frac{V_t}{(1+r)^t}\Big) \Big]  \nonumber \\
    &\eqqcolon{E}_{Q} \left[ \esssupN_{\theta \in \Theta} G(\omega, a,b,c,d) \right], \label{eq:NewCalculations2}
 \end{align}
For fixed $\omega \in \Omega$ the function $G_{\omega}:=G(\omega, \cdot):\Theta \to \mathbb{R}$ is decreasing in $b$ and $c$ and thus \eqref{eq:NewCalculations2} falls to 
 \begin{align} \label{eq:NewCalculations3}
     E_{Q \odot \ccP_{\Theta}}[X]={E}_{Q} \Big[ \esssupN_{(a,d)\in A\times D} G(\omega, a,\underline{b},\underline{c},d) \Big].
 \end{align}
\raggedbottom

Moreover, we want to compare the robust price in \eqref{eq:NewCalculations3} with the supremum over all classical $Q P^{\theta}$-evaluations for $\theta \in \Theta$, which is given by 
\begin{equation} \label{eq:SupremumOverNoRobustPrice}
\sup_{\theta \in \Theta} E_{Q \odot P^{\theta}} [X]=\sup_{\theta \in \Theta} E_{Q}\left[ E_{P^{\theta}}[X \vert \ccF] \right] = \sup_{(a,d) \in A \times D} {E}_{Q} \left[ G(\omega, a,\underline{b},\underline{c},d) \right].
\end{equation}

For the numerical evaluation we consider the time horizon $T=8$ and use the following parameters. In the CRR model we set $S_0^1=100,$ $r=0.05$, $u=0.1$, and $v=-0.1$. Furthermore, we fix the strike $K=100$, the penalty parameter $l=0.1$, and the guaranteed interest rate $r_G=0.01$. The parameter set $\Theta$ is given by $\Theta:=\left[ 50,340\right] \times \left[ 0.02,0.03\right] \times \left[ 0.01,0.05\right] \times \left[ 10^4, 10^5\right].$ To calculate the value of the robust $Q  \ccP_{\Theta}$-evaluation in \eqref{eq:NewCalculations3}, we maximize for each path $\omega \in \Omega$ the function $G_{\omega}$ over $\Theta$ by using the Nelder-Mead method. %
As already noted above, due to the monotonicity properties of $G_{\omega}$ with respect to $b$ and $c$, we only optimize over $(a,d) \in A \times D.$ While the maximum of $G_{\omega}$ for the parameter $d$ is numerically always attained at the upper boundary of $D$, the optimal value for $a$ varies. In particular, we get the following numerical results for $E_{Q \odot \ccP^{\Theta}}[X]$ and $\sup_{\theta \in \Theta} E_{Q \odot P^{\theta}} [X]$:
\begin{align}
    E_{Q \odot \ccP_{\Theta}}[X]=88.38 \quad \text{ and } \quad \sup_{\theta \in \Theta} E_{Q \odot P^{\theta}} [X]=87.61,
\end{align}
with the difference $\Delta:=E_{Q \odot \ccP_{\Theta}}[X]-\sup_{\theta \in \Theta} E_{Q \odot P^{\theta}} [X]$  given by $\Delta=0.67.$
These results indicate that the robust price, which reflects the model risk and which guarantees that the finance-insurance market is arbitrage-free, is strictly greater than the worst case price of all possible models. This is illustrated in Figure \nolinebreak \ref{fig1}.

	\begin{figure}[ht]
		\centering
		\includegraphics[width=0.6\textwidth]{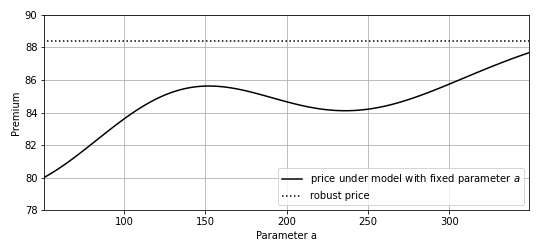}
		\caption{Asymptotic insurance-finance arbitrage-free premium for $a\in[50,350]$ and fixed parameters $\underline{b},$ $\underline{c}$ and  $\overline{d}$.}
		\label{fig1}
	\end{figure}
\end{example}

\subsection{Modelling conditional dependence with copulas}\label{sec:CoxModel}
In Section \ref{sec:conditionallyindependence} we assumed that the random times $\tau^1$ and $\tau^2$ are $\ccF$-conditionally independent, see Equation \eqref{eq:conditionalindependenceoftau}. We now drop this assumption and allow some dependence between $\tau^1$ and $\tau^2$ which is modelled by a family of copulas $(C_{\theta})_{\theta\in \Theta}$. 
We assume that under each $P^{\theta} \in \ccP_{\Theta}$ the law of $(\tau^1,\tau^2)$ is described by its marginal distribution functions $F_1^{\ccF}(\theta,\cdot)$ and $F_2^{\ccF}(\theta,\cdot)$ and the copula $C^{\theta}:[0,1]^2  \to [0,1] $ such that
\begin{equation} \label{eq:ConditionalDistributionCopulaGeneral}
    P^{\theta}[\tau^1 \leq s, \tau^2 \leq t \vert \ccF]= C^{\theta}  (F_1^{\ccF}(\theta,s), F_2^{\ccF}(\theta,t)) \quad \text{for } s,t \in \mathbb{N}_0.
\end{equation}
\raggedbottom
We introduce the following notation for the marginal conditional survival probabilities:
\begin{equation*}
\bar{F}_{1}^{\ccF}(\theta,t):=1-{F}_{1}^{\ccF}(\theta,t) \quad \text{ and } \quad  \bar{F}_{2}^{\ccF}(\theta,t):=1-{F}_{2}^{\ccF}(\theta,t) \quad \text{for } t \in \mathbb{N}_0.
\end{equation*}
The conditional survival function of $(\tau^1,\tau^2)$ is then given by
\begin{align} \label{eq:SurvivalConditionalFunction}
    P^{\theta}[\tau^1 > s, \tau^2 > t \vert \ccF]=\bar{F}_1^{\ccF}(\theta,s) +\bar{F}_2^{\ccF}(\theta,t)-1 + C^{\theta}(1-\bar{F}_1^{\ccF}(\theta,s),1- \bar{F}_2^{\ccF}(\theta,t)),
\end{align}
which can be described by $\hat{C}^{\theta}:\Theta\times[0,1]^2  \to [0,1] $ defined by
\begin{equation*}
    \hat{C}^{\theta}(u,v):= u+v-1 +C^{\theta}(1-u,1-v),
\end{equation*}
such that \eqref{eq:SurvivalConditionalFunction} can be rewritten in 
\begin{equation*}
    P^{\theta}[\tau^1 > s, \tau^2 > t \vert \ccF]=\hat{C}^{\theta}(\bar{F}_1^{\ccF}(\theta,s), \bar{F}_2^{\ccF}(\theta,t)).
\end{equation*}
Note that $\hat{C}^{\theta}$ is also a copula, cf.\@ \cite[Section 2.6]{nelsen2007introduction}. 
\begin{example}
A suitable class of copulas $C^{\theta}$ that depends on only one parameter is the class of Archimedean copulas (see\@ \cite[Section 4.2]{nelsen2007introduction} or \cite{schmidt2007coping} for further details and related literature). Most of these copulas can be represented by an explicit formula, which leads to a high tractability for applications. The copula $C^{\theta}(u^1,u^2)=C(u^1, u^2):=u^1 u^2$ reflects the setting of Section \ref{sec:conditionallyindependence}, cf.\@  Equation \eqref{eq:conditionalindependenceoftau}.
\end{example}
\raggedbottom
Let us consider the insurance benefit $X$ given in \eqref{eq:benefitX} and evaluate this benefit with the $Q\ccP_{\Theta}$-rule. For $Q \in \mathcal{M}_e(\mathbb{F})$ we get
\allowdisplaybreaks
\begin{align}
    &E_{Q \odot \ccP_{\Theta}}\left[ X\right] \nonumber  \\
    &=E_{Q \odot \ccP_{\Theta}}\Big[ \textbf{1}_{\lbrace \tau^1 >T, \tau^2 >T\rbrace} Y^1(1+r)^{-T}+\sum_{t=1}^{T-1} \textbf{1}_{\lbrace \tau^1 >t, \tau^2 =t\rbrace} Y^2_t(1+r)^{-t}\Big]  \nonumber \\
    &= E_{Q} \Big[ \esssupN_{P^{\theta} \in \ccP_{\Theta}} E_{P^{\theta}}\Big[\textbf{1}_{\lbrace \tau^1 >T, \tau^2 >T\rbrace} Y^1 (1+r)^{-T} + \sum_{t=1}^{T-1} \textbf{1}_{\lbrace \tau^1 >t, \tau^2 =t\rbrace} Y^2_t (1+r)^{-t} \big \vert \ccF\Big] \Big] \nonumber \\
    &= E_{Q} \Big[ \esssupN_{P^{\theta} \in \ccP_{\Theta}} \Big(Y^1 (1+r)^{-T} P^{\theta}[ \tau^1 >T, \tau^2 >T  \vert \ccF] + \sum_{t=1}^{T-1} Y_t^2 (1+r)^{-t} P^{\theta}[\tau^1 >t, \tau^2 =t \vert \ccF]\Big) \Big] \nonumber \\
    &= E_{Q} \Big[ \esssupN_{P^{\theta} \in \ccP_{\Theta}} \Big( Y^1 (1+r)^{-T} \hat{C}^{\theta}(\bar{F}_1^{\ccF}(\theta,T), \bar{F}_2^{\ccF}(\theta,T)) \nonumber \\
    & \quad +\sum_{t=1}^{T-1}Y_t^2  (1+r)^{-t} \left( \hat{C}^{\theta}\left( \bar{F}_1^{\ccF}(\theta,t), \bar{F}_2^{\ccF}(\theta,t-1) \right) - \hat{C}^{\theta}\left( \bar{F}_1^{\ccF}(\theta,t), \bar{F}_2^{\ccF}(\theta,t) \right) \right) \Big) \Big],\label{QP_Theta_valuation_dependence} 
\end{align}
where we use in \eqref{QP_Theta_valuation_dependence} that    
\begin{align*}
        P^{\theta}\left[ \tau^1 >t, \tau^2=t \vert \ccF\right]&= P^{\theta}\left[ \tau^1 >t, \tau^2>t-1 \vert \ccF\right]- P^{\theta}\left[ \tau^1 >t, \tau^2>t \vert \ccF\right] \\
        &=\hat{C}^{\theta}\left( \bar{F}_1^{\ccF}(\theta,t), \bar{F}_2^{\ccF}(\theta,t-1) \right) - \hat{C}^{\theta}\left( \bar{F}_1^{\ccF}(\theta,t), \bar{F}_2^{\ccF}(\theta,t) \right).
  \end{align*}

\begin{example} \label{example:CoxModel}
We briefly show that the Cox model fits in the framework introduced above. 
To do this, let $\Lambda^i=(\Lambda^i_t)_{t=0,...,T}$ with $\Lambda^i_0=0$ for $i=1,2$ be two increasing processes on $(\Omega, \ccF, P_0)$ and $E_1$ and $E_2$ two random variables on $(\Omega, \ccG).$ For the parameter set
$$
\Theta:=\Gamma^1 \times \Gamma^2 \times \Xi :=\big[ \underline{\gamma^1}, \overline{\gamma^1}\big] \times \big[ \underline{\gamma^2}, \overline{\gamma^2}\big] \times \big[ \underline{\xi}, \overline{\xi}\big] \subseteq \mathbb{R}^2_{>0} \times \mathbb{R},
$$
we choose $\ccP_{\Theta} \subseteq \mathfrak{P}_{\mathcal{N}}(\Omega, \ccG)$ such that for any $P^{\theta} \in \ccP_{\Theta}$ with $\theta=(\gamma^1,\gamma^2, \xi) \in \Theta$ it holds that
$$
E^1 \sim^{P^{\theta}} \exp(\gamma^1) \quad \text{ and }\quad E^2 \sim^{P^{\theta}} \exp(\gamma^2).
$$
Moreover, $E^1$ and $E^2$ are assumed to be conditionally independent of $\ccF$ under each $P_{\theta}\in\ccP_{\Theta}$. We define the default times $\tau^1$ and $\tau^2$ using
    \begin{align*}
        \tau^i:=\inf \lbrace t \in \mathbb{N}_0 \ \vert \ \Lambda_t^i \geq E^i\rbrace, \quad i=1,2,
    \end{align*}
    where we use the convention $\inf \emptyset:=\infty.$
In this case $F_1^{\ccF}(\theta,t)$ and $F_2^{\ccF}(\theta,t)$ are given by
$$
F_1^{\ccF}(\theta,t)= 1- e^{-\gamma^1 \Lambda_t^1} \quad \text{ and }\quad F_2^{\ccF}(\theta,t)= 1- e^{-\gamma^2 \Lambda_t^2} \quad \text{ for } t \in \mathbb{N}_0.
$$
\end{example}
\begin{remark}
Note that a Cox model with a single default time under uncertainty is studied in  \cite{biagini_zhang_2019}. In contrast to Example \ref{example:CoxModel}, it is assumed that that there is some uncertainty on the financial market, i.e., a set $\ccP \subseteq \mathfrak{P}(\Omega, \ccF)$ is considered, but the law of the random variable $E$ is fixed. 
\end{remark}

\section*{Conclusion}
In this paper we introduce a definition of robust asymptotic insurance-finance arbitrage and characterize the absence of arbitrage. Our main theorem provides an economic foundation for the pricing of finance-linked insurance products via $Q\ccP$-evaluations. Moreover, we show that valuation principles from a specific class of two-step evaluations can be written as $Q\ccP$-evaluation.  We also apply our results in such a way as to model an insurance-finance market, before  concluding with some numerical observations. 

\appendix

\section{Proofs of Section \ref{sec:TwoStepValuation}} \label{appendixProofs}
\begin{proposition}\label{prop:law}
The law of $(\pi_{\Omega^F},\pi_{\Omega^I_j})$ under $\mu_P$ equals $P$ for all $j\in\mathbb N$.
Moreover, given any $j\in\mathbb N$, the law of $(S,X^j)$ under $\mu_P$ equals the law of $(\tilde S,\tilde X)$ under $P$. 
\end{proposition}
\begin{proof}
Let $A\in\mathscr F^F$ and $C\in \mathscr F^I$. Then it holds that 
\begin{align*}
   \mu_P^{(\pi_{\Omega^F},\pi_{\Omega^I_j})}(A\times C)&=\int\mathds{1}_{\{\tilde\pi_{\Omega^F}\in A\}}(P^{\tilde\pi_{\Omega^I}\vert \tilde{\mathscr F}})^{\otimes \mathbb N}(\Omega^I\times\cdots\times \Omega^I\times C\times\Omega^I\times\cdots)dP \\
   &=\int\mathds{1}_{\{\tilde\pi_{\Omega^F}\in A\}}P^{\tilde\pi_{\Omega^I}\vert \tilde{\mathscr F}}(C)dP\\ 
   &=P(A\times C)\quad \text{for all $j\in\mathbb N$, }
\end{align*}
where we use the definition of the conditional probability and the tower property in the third equality. 
The second statement then follows by
\begin{align*}
\mu_P^{(S,X^j)}&=\left(\mu_P^{(\pi_{\Omega^F},\pi_{\Omega^I_j})}\right)^{(\tilde S,\tilde X)}
    = P^{(\tilde S,\tilde X)}\quad \text{for all $j\in\mathbb N$. }\qedhere
\end{align*}
\end{proof}

\begin{proposition}\label{prop:conditionalexpectation}
Let $j\in\mathbb N$ and $\nu_P$ be a measure on $(\Omega,\mathscr G)$ which fulfills \eqref{eq:mupcondition1}, i.e.,
\begin{align*}
    \nu_P^{(\pi_{\Omega^F},\pi_{\Omega^I_j})}=P \quad\text{for all }j\in\mathbb N.
\end{align*}
Then it holds for all $A\in \mathscr F^I$ that
\begin{align}\label{eq:conditionalexpectationmup}
E_{\nu_P}[\mathds{1}_{\{\pi_{\Omega^I_j}\in A\}}\vert \mathscr F]=E_{P}[\mathds{1}_{\{\tilde\pi_{\Omega^I}\in A\}}\vert \tilde{\mathscr F}]\circ (\pi_{\Omega^F},\pi_{\Omega^I_\ell})\quad\text{ $\nu_P$-a.s.\@ for all $A\in\mathscr F^I$ and $\ell\in\mathbb N$.}
\end{align}
\end{proposition}
\begin{proof}
Let $B\in\mathscr F$ and $A\in\mathscr F^I$. Then we have 
\begin{align*}
    \int\mathds{1}_B\mathds{1}_{\lbrace \pi_{  \Omega^I_j \in A }\rbrace}d\nu_P&=\int(\mathds{1}_{\pi_{\Omega^F}(B)}\circ \pi_{\Omega_F})(\mathds{1}_{A}\circ \pi_{\Omega^I_j})d\nu_P\\
    &=\int(\mathds{1}_{\pi_{\Omega^F}(B)}\circ \tilde\pi_{ \Omega_F})(\mathds{1}_{A}\circ \tilde \pi_{\Omega^I})dP\\
    &=\int(\mathds{1}_{\pi_{\Omega^F}(B)}\circ \tilde \pi_{ \Omega_F})E_{P}[(\mathds{1}_{A}\circ \tilde \pi_{\Omega^I})\vert\tilde{\mathscr F}]dP\\
    &=\int\mathds{1}_B E_{P}[(\mathds{1}_{A}\circ \tilde \pi_{\Omega^I})\vert\tilde{\mathscr F}]\circ ((\pi_{\Omega^F},\pi_{\Omega^I_\ell})) d\nu_P,
\end{align*}
where we use \eqref{eq:mupcondition1} in the second equality and $\tilde{\pi}^{-1}_{\Omega_F}(\pi_{\Omega^F}(F))\in \tilde{\mathscr F}$ in the third equality.
\end{proof}
\begin{proposition}\label{prop:conditionalindependence}
The projections $(\pi_{\Omega^I_j})_{j\in\mathbb N}$ are $\mathscr F$-conditionally \emph{iid} under $\mu_P$ and consequently also the insurance benefits $(X^j)_{j\in\mathbb N}$ are $\mathscr F$-conditionally \emph{iid} under $\mu_P$.
\end{proposition}
\begin{proof}
We  must show that for every finite subset $J\subset \mathbb N$ and $A_j\in \mathscr F^I$ for all $j\in J$ that
\begin{align*}
    E_{\mu_P}\Big[\prod_{j\in J}\mathds{1}_{\{\pi_{\Omega^I_j}\in A_j\}}\Big\vert \mathscr F\Big] = \prod_{j\in J} E_{\mu_P}\Big[\mathds{1}_{\{\pi_{\Omega^I_j}\in A_j\}}\Big\vert \mathscr F\Big].
\end{align*}
This follows because for every $B\in\mathscr F$, it holds that
\begin{align*}
    E_{\mu_P}\Big[\mathds{1}_B\prod_{j\in J}\mathds{1}_{\{\pi_{\Omega^I_j}\in A_j\}}\Big]&=E_{P}\Big[\mathds{1}_{\tilde\pi^{-1}_{\Omega^F}(\pi_{\Omega_F}(B))}\prod_{j\in J}P^{\tilde\pi_{\Omega^I}\vert\tilde{\mathscr F}}(A_j)\Big]\\
    &=E_{P}\Big[\mathds{1}_{\tilde\pi^{-1}_{\Omega^F}(\pi_{\Omega_F}(B))}\prod_{j\in J}E_{P}[\mathds{1}_{\{\tilde\pi_{\Omega^I}\in A_j\}}\vert \tilde{\mathscr F}]\Big]\\
    &=E_{\mu_P}\Big[\mathds{1}_B\prod_{j\in J}E_{P}[\mathds{1}_{\{\tilde\pi_{\Omega^I}\in A_j\}}\vert \tilde{\mathscr F}]]\circ (\pi_{\Omega^F},\pi_{\Omega^I_1})\Big]\\
    &=\int\mathds{1}_B\prod_{j\in J}E_{\mu_p}[\mathds{1}_{\{\pi_{\Omega^I_j}\in A_j\}}\vert \mathscr F]d\mu_P,
\end{align*}
where we use the definition of $\mu_P$ given in \eqref{eq:disintegration} in the first equality and Proposition \ref{prop:conditionalexpectation} in the fourth equality. 
\end{proof}

\begin{proposition} \label{eq:ProofUniqueness}
The measure $\mu_P$ defined by \eqref{eq:disintegration} is the unique measure on $(\Omega,\mathscr G)$ which fulfills \eqref{eq:mupcondition1} and \eqref{eq:mupcondition2}.
\end{proposition}
\begin{proof}
Using Proposition \ref{prop:law} and Proposition \ref{prop:conditionalindependence} the measure $\mu_P$ fulfills \eqref{eq:mupcondition1} and \eqref{eq:mupcondition2}. Moreover, let $\nu_P$ be a measure on $(\Omega,\mathscr G)$ which also satisfies \eqref{eq:mupcondition1} and \eqref{eq:mupcondition2} (where $\mu_P$ is replaced by $\nu_P$).  Let $J\subset \mathbb N$ be a finite subset and $A\in \mathscr F^F$ and $C_j\in\mathscr F^I$ for all $j\in J$. Then for $D$ given by
\begin{align}\label{eq:nstablegeneratorsets}
    D=\pi^{-1}_{\Omega^F}(A)\cap \bigcap_{j\in J}\pi^{-1}_{\Omega^I_j}(C_j)
\end{align}
we get the following:
\begin{align*}
    \nu_P(D)&=E_{\nu_P}\Big[\mathds{1}_{\pi^{-1}_{\Omega^F}(A)}E_{\nu_P}\Big[\prod_{j\in J}\mathds{1}_{\pi_{\Omega_j^I}\in C_j}\vert \mathscr F\Big]\Big]\\
    &=E_{\nu_P}\Big[\mathds{1}_{\pi^{-1}_{\Omega^F}(A)}\prod_{j\in J}E_{\nu_P}\Big[\mathds{1}_{\pi_{\Omega_j^I}\in C_j}\vert \mathscr F\Big]\Big]\\
    &=E_{\nu_P}\Big[\mathds{1}_{\pi^{-1}_{\Omega^F}(A)}\prod_{j\in J}E_{P}\Big[\mathds{1}_{\tilde\pi_{\Omega^I}\in C_j}\vert \tilde{\mathscr F}\Big]\circ (\pi_{\Omega^F},\pi_{\Omega^I_1})\Big]\\
    &=E_{P}\Big[\mathds{1}_{\tilde\pi^{-1}_{\Omega^F}(A)}\prod_{j\in J}E_{P}\Big[\mathds{1}_{\tilde\pi_{\Omega^I}\in C_j}\vert \tilde{\mathscr F}\Big]\Big],\\
\end{align*}
where we use \eqref{eq:mupcondition2} in the third equality and Proposition \ref{prop:conditionalexpectation} in the fourth equality. The same calculations can be done for $\mu_P$. The sets in \eqref{eq:nstablegeneratorsets} form a $\cap$-stable generator of $\mathscr G$ and thus we obtain $\mu_P=\nu_P$. This demonstrates the uniqueness.   
\end{proof}
\begin{proposition}\label{prop:martingalemeasure}
The set $\mathcal M_e(\mathbb F)$ of all measures on $(\Omega,\mathscr F)$ such that $S$ is a $\mathbb F$-martingale and which are equivalent to $\mu_{P_0}$ 
is given by
\begin{align*}
    \mathcal M_e(\mathbb F)=\{\mu_{ Q}\vert  Q\in\mathcal M_e(\tilde{\mathbb F})\}.
\end{align*}
\end{proposition}
\begin{proof}
We show that for all $t,s \in \lbrace 0,...,T\rbrace$ it holds that
\begin{align} \label{eq:MartingaleMeasuresProof1}
 E_{\mu_{Q}}[S_t \vert \ccF_s]= E_{Q}[\tilde{S}_t \vert \tilde{\ccF}_s] \circ (\pi_{\Omega^F}, \pi_{\Omega_1^I}).  
 \end{align}
 Let $B_s \in \ccF_s,$ then we have
 \begin{align}\label{eq:MartingaleMeasuresProof2}
 \begin{split}
      E_{\mu_{Q}}[\mathds{1}_{B_s}S_t]&=  E_{\mu_{Q}}\left[\mathds{1}_{B_s} \tilde{S}_t \circ (\pi_{\Omega^F}, \pi_{\Omega_1^I})\right] \\
 &=E_{Q}\left[\mathds{1}_{\tilde{\pi}^{-1}_{\Omega^F}(\pi_{\Omega^F}(B_s))} \tilde{S}_t\right] \\
  &=E_{Q}\left[\mathds{1}_{\tilde{\pi}^{-1}_{\Omega^F}(\pi_{\Omega^F}(B_s))} E_{Q}[\tilde{S}_t\vert \tilde{\ccF}_s]\right] \\
  &=E_{\mu_{Q}} \left[ \mathds{1}_{B_s} E_{Q}[\tilde{S}_t \vert \tilde{\ccF}_s] \circ (\pi_{\Omega^F}, \pi_{\Omega_1^I})\right] 
 \end{split}
 \end{align}
 Thus, it follows \eqref{eq:MartingaleMeasuresProof1} and so we can conclude that $Q \in \mathcal{M}_e(\tilde{\mathbb{F}})$ implies $\mu_{Q} \in \mathcal{M}_e(\mathbb{F}).$ The equivalence of $\mu_{P_0}$ and $\mu_{ Q}$ follows from the equivalence of $P_0$ and $ Q$.  For the other inclusion, let $R \in \mathcal{M}_e(\mathbb{F})$. We must now  show that there exists $Q \in \mathcal{M}_e(\tilde{\mathbb{F}})$ such that $R= \mu_{Q}$. Define $Q$ by
 $$
 Q:=R^{\pi_{\Omega^F}}(\tilde{\pi}_{\Omega^F}(\cdot)). 
 $$
 Then, by changing the roles of $\pi_{\Omega^F}$ and $\tilde{\pi}_{\Omega^F}$ the result follows using \eqref{eq:MartingaleMeasuresProof2}.
 \end{proof}
\begin{proposition}\label{prop:premium}
For any $P \in \ccP$ and $ Q\in \mathcal M_e(\tilde{\mathbb F})$ we have the following:
\begin{align*}
E_{ Q}\Big[E_{P}[\tilde X\vert \tilde{\mathscr F}] \Big]&=E_{\mu_{ Q}}[E_{\mu_P}[X^1\vert\mathscr F] ],\\
E_{ Q}\Big[\esssupN_{P \in \ccP} E_{P}[\tilde X\vert \tilde{\mathscr F}] \Big]&=E_{\mu_{ Q}}\left[\esssupN_{P \in \ccP}E_{\mu_P}[X^1\vert\mathscr F] \right].
\end{align*}
\end{proposition}
\begin{proof}
By using the same arguments as in the proof of Proposition \ref{prop:martingalemeasuretransformation}, we find that
\begin{align*}
    E_{\mu_P}[X^1\vert \mathscr F]= E_P[\tilde X\vert \tilde{\mathscr F}]\circ (\pi_{\Omega^F},\pi_{\Omega^I_1}).
\end{align*}
This implies that 
\begin{align*}
    E_{\mu_{ Q}}\Big[E_{\mu_P}[X^1\vert \mathscr F]\Big]&=E_{\mu_{ Q}}\Big[E_P[\tilde X\vert \tilde{\mathscr F}]\circ (\pi_{\Omega^F},\pi_{\Omega^I_1})\Big]\\
    &=E_{ Q}\Big[E_{P}[\tilde X\vert \tilde{\mathscr F}] \Big].
\end{align*}
The second statement follows analogously. 
\end{proof}
\bibliographystyle{agsm}
\bibliography{sample}

\end{document}